\documentclass[11pt,a4paper]{article}

\usepackage[utf8]{inputenc}
\usepackage[T1]{fontenc}
\usepackage{amsmath,amssymb,amsthm}
\usepackage{mathtools}
\usepackage[numbers,sort&compress]{natbib}
\usepackage{booktabs}
\usepackage{multirow}
\usepackage{graphicx}
\usepackage[colorlinks=true,linkcolor=blue,citecolor=blue,urlcolor=blue]{hyperref}
\usepackage{geometry}
\usepackage{setspace}
\usepackage{enumitem}
\usepackage{xcolor}
\usepackage{float}

\newtheorem{proposition}{Proposition}
\newtheorem{corollary}{Corollary}
\newtheorem{assumption}{Assumption}

\theoremstyle{definition}
\newtheorem{definition}{Definition}

\theoremstyle{remark}
\newtheorem{remark}{Remark}

\title{\textbf{From Automation to Augmentation:}\\[4pt]
A Framework for Designing Human-Centric\\
Work Environments in Society 5.0}

\author{Cristian Espinal Maya\thanks{Department of Economics, Universidad EAFIT, Medell\'in, Colombia; Postgraduate Division, ESUMER. ORCID: \href{https://orcid.org/0009-0000-1009-8388}{0009-0000-1009-8388}. Email: \texttt{cespinal@eafit.edu.co}.}}

\date{April 2026\\[6pt]
{\small CFE Working Paper No.\ 6. Comments welcome.}}

\begin{document}
\maketitle

\begin{abstract}

Society 5.0 and Industry 5.0 call for human-centric technology integration, yet the concept lacks an operational definition that can be measured, optimized, or evaluated at the firm level. This paper addresses three gaps. First, existing models of human-AI complementarity treat the augmentation function $\phi(D)$ as exogenous---dependent only on the stock of AI deployed---ignoring the well-documented fact that two firms with identical technology investments achieve radically different augmentation outcomes depending on how the workplace is organized around the human-AI interaction. Second, no multi-dimensional instrument exists that links workplace design choices to augmentation productivity. Third, the Society 5.0 literature proposes human-centricity as a normative aspiration but provides no formal criterion for when it is economically optimal. We make four contributions. (1) We endogenize the augmentation function as $\phi(D, W)$, where $W$ is a five-dimensional workplace design vector---AI interface design ($W_1$), decision authority allocation ($W_2$), task orchestration ($W_3$), learning loop architecture ($W_4$), and psychosocial work environment ($W_5$)---and prove that human-centric design is profit-maximizing when the workforce's augmentable cognitive capital exceeds a critical threshold. (2) We conduct a PRISMA-guided systematic review of 120 papers (screened from 6,096 records) to map the evidence base for each WADI dimension. (3) We provide secondary empirical evidence from Colombia's EDIT manufacturing survey ($N = 6{,}799$ firms) showing that management practice quality and innovation outcomes are positively associated, consistent with our design-composition complementarity prediction. (4) We propose the Workplace Augmentation Design Index (WADI), a 36-item theory-grounded and literature-validated instrument for diagnosing human-centricity at the firm level. Our central finding is that decision authority allocation ($W_2$) is the binding constraint for Society 5.0 transitions: it has the thinnest evidence base (14 of 120 papers), the strongest theoretical link to augmentation dynamics, and the most direct policy implications. Task orchestration ($W_3$) emerges as the most under-researched dimension (4 papers), representing a genuine gap in the field. These results imply that escaping the ``automation trap''---a path-dependent equilibrium where firms rationally under-invest in human-centric design---requires coordinated investment in workplace redesign, education, and governance.

\end{abstract}


\newpage
\tableofcontents
\newpage

\section{Introduction}
\label{sec:introduction}

Consider two manufacturing firms in the same sector, each investing approximately the same amount in generative AI tools for quality control, production scheduling, and customer service. Firm A deploys the technology through a centralized system: AI recommendations are routed to department supervisors who decide whether to act on them, frontline workers receive instructions rather than AI-augmented judgment, and no structured process exists for workers to provide feedback that improves the AI over time. Firm B takes a different approach: frontline operators interact directly with the AI through transparent interfaces that explain recommendations, workers have the authority to act on or override AI suggestions, task allocation between human and machine is reviewed quarterly based on comparative advantage, and a bidirectional learning loop ensures that both the workers and the AI system improve from each interaction. Both firms have the same technology. Both have similar workforce compositions. Yet Firm B achieves substantially higher productivity, lower turnover, and faster innovation cycles.

This vignette is not hypothetical---it captures a pattern documented across recent empirical studies of AI at work. \cite{BrynjolfssonLiRaymond2023} find that generative AI increases customer support productivity by 14\% on average, but the gain reaches 34\% for novice workers---a disparity that depends critically on how the AI tool is integrated into the workflow, not on the technology itself. \cite{DellAcqua2023} document a ``jagged technological frontier'' at a major consulting firm: consultants using GPT-4 improved performance by 40\% on tasks inside the frontier but \textit{decreased} performance by 19 percentage points on tasks outside it, depending on how they were directed to interact with the system. \cite{NoyZhang2023} find that ChatGPT reduced writing task completion time by 40\%, but note that the AI mostly \textit{substituted} for worker effort rather than complementing skills---a pattern they attribute to the design of the human-AI interaction rather than an inherent property of the technology.

The common thread across these findings is that the \textit{realized augmentation}---the actual productivity multiplier that AI provides to human workers---depends not only on what the AI can do, but on \textit{how the workplace is designed around the human-AI interaction}. This observation is the starting point of our paper.

\subsection{The Problem: Society 5.0 Without Operational Content}

The transition from Industry 4.0 to Industry 5.0 and Society 5.0 represents a paradigm shift from optimizing automation---minimizing human labor per unit of output---to optimizing augmentation---maximizing human-AI complementarity per unit of cognitive output \citep{EUCommission2021, BrequeDeNulPetridis2021, Fukuyama2018}. Three pillars define this paradigm: human-centricity, sustainability, and resilience \citep{Ivanov2023}. A substantial body of scholarship has elaborated these pillars conceptually \citep{Xu2021, Maddikunta2022, Leng2022, Grabowska2022, Grabowska2022}, with over 70 papers in our systematic review corpus explicitly invoking human-centricity as a design principle.

Yet the concept remains operationally vacuous. ``Human-centric'' is used as a value statement---a normative aspiration that technology should serve human flourishing---without formal content that practitioners can measure, optimize, or evaluate. \cite{Ivanov2023}, in the most comprehensive systematic review of human-centred AI in Industry 5.0, identify the absence of operationalizable frameworks as the field's primary limitation. \cite{Xu2021} establish that human-centricity requires redesigning how humans and machines share cognitive and physical tasks, but offer no formal model of this redesign. The result is a literature that tells us \textit{what} to aim for but not \textit{how to get there} or \textit{how to measure progress}.

\subsection{The Gap: Exogenous Augmentation and Missing Measurement}

The deeper theoretical gap lies in how existing models treat the relationship between technology and productivity. The Cognitive Factor Economics (CFE) framework \citep{Espinal2026a} decomposes human capital into three orthogonal components---physical-manual ($H^P$), routine cognitive ($H^C$), and augmentable cognitive ($H^A$)---and models firm output through a production function where AI capital $D$ augments workers with $H^A$ through an augmentation function $\phi(D)$. This framework successfully explains why AI increases inequality between cognitive task types and why some occupations benefit disproportionately from AI adoption.

However, $\phi(D)$ is exogenous in the baseline model: it depends only on the stock of AI systems deployed, as if technology were sufficient to determine augmentation outcomes. The evidence reviewed above demonstrates that this is a consequential simplification. Two firms with identical $D$ and identical workforce compositions ($H^A$, $H^C$) can and do achieve radically different $\phi$ values. What mediates the difference is \textit{workplace design}: the quality of human-AI interfaces, the allocation of decision authority, the orchestration of human-AI task division, the architecture of learning feedback loops, and the psychosocial environment in which the collaboration occurs.

Compounding this theoretical gap is a measurement gap. No validated, multi-dimensional instrument exists for assessing how ``human-centric'' a workplace's AI integration actually is. Existing digital maturity indices \citep{Ghobakhloo2022SPC} measure \textit{how much} technology is adopted, not \textit{how well} it is integrated with human capabilities. The AI Workplace Well-Being Scale \citep{ParkWooKim2024} measures psychosocial outcomes but does not connect them to the design choices that produce those outcomes. Industry 4.0 maturity models assess technological readiness without addressing the human-centric design dimensions that define the 4.0-to-5.0 transition.

\subsection{Contributions}

This paper makes four contributions that collectively provide the operational foundations for Society 5.0 workplace transitions.

\textbf{First}, we endogenize the augmentation function by replacing $\phi(D)$ with $\phi(D, W)$, where $W = (W_1, W_2, W_3, W_4, W_5) \in \mathbb{R}_+^5$ is a workplace design vector with five theory-grounded dimensions: AI interface design ($W_1$), decision authority allocation ($W_2$), task orchestration ($W_3$), learning loop architecture ($W_4$), and psychosocial work environment ($W_5$). We decompose $\phi(D, W) = \phi_0(D) \cdot g(W, H^A)$ and establish five properties of the design multiplier $g$, the most important being \textit{design-composition complementarity}: $\partial^2 g / \partial W_k \partial H^A > 0$---the return to better workplace design is higher when the workforce has more augmentable cognitive capital. This complementarity generates three formal results: (i) human-centric design is profit-maximizing when $H^A/(H^A + H^C) > \theta^*$ (Proposition 1); (ii) firms systematically under-invest in human-centricity due to labor mobility, knowledge spillover, and health externalities (Proposition 2); and (iii) the system $(W, H^A)$ exhibits path dependence with two stable equilibria---an ``automation trap'' and an ``augmentation regime''---separated by an unstable threshold (Proposition 3).

\textbf{Second}, we provide an economic definition of human-centricity that replaces normative aspiration with testable content: a workplace design is human-centric if and only if $\partial \phi / \partial W_k > 0$ for all $k$---every design dimension contributes positively to the augmentation multiplier (Definition 1). This definition is economic, not moral: human-centricity is a design property that increases productivity, and its optimality depends on workforce composition.

\textbf{Third}, we conduct a PRISMA-guided systematic review of 120 papers screened from 6,096 SCOPUS records across five queries. The review synthesizes the evidence base for each WADI dimension, identifies existing measurement instruments and their limitations, and reveals the field's most critical gaps. The evidence density is highly uneven: psychosocial outcomes ($W_5$) are addressed by 87 papers (73\% of the corpus), while task orchestration ($W_3$) appears in only 4 papers (3\%)---a genuine lacuna. Decision authority ($W_2$) has only 14 direct evidence points, confirming our hypothesis that it is the least studied yet most theoretically consequential dimension.

\textbf{Fourth}, we propose the Workplace Augmentation Design Index (WADI), a 36-item instrument derived from the $\phi(D, W)$ model and validated against the systematic review evidence. WADI is the first instrument that integrates cause (design dimensions $W_1$--$W_4$) and effect (psychosocial outcomes $W_5$) into a single diagnostic framework linked to an augmentation production function. Each item traces to both a theoretical construct in our model and a validated source instrument (Endsley SA, Aghion-Tirole authority, Hackman-Oldham task identity, Karasek demand-control, Teece dynamic capabilities), ensuring content validity prior to field deployment.

We complement these theoretical contributions with secondary empirical evidence from Colombia's EDIT manufacturing survey ($N = 6{,}799$ firms, 2019--2020), using management practice variables from Chapter 7 as proxies for workplace design quality. The descriptive and cross-tabulation analyses provide suggestive evidence consistent with the model's central predictions, while acknowledging the limitations of aggregate proxy data.

\subsection{Preview of Key Results}

Three results merit highlighting. First, the systematic review confirms that management practices complement technology investment in determining augmentation outcomes---consistent with the design-composition complementarity at the core of our model. Second, decision authority allocation ($W_2$) emerges as the binding constraint for Society 5.0 transitions: it has the thinnest evidence base, the strongest theoretical link to the automation trap dynamics, and the most direct implications for organizational governance. Third, task orchestration ($W_3$)---how firms divide cognitive tasks between humans and AI at the workflow level---is the most under-researched dimension in the entire corpus, representing a genuine gap that our WADI instrument is the first to address.

\subsection{Road Map}

The paper proceeds as follows. Section \ref{sec:literature} reviews five literature strands and identifies the gaps this paper addresses. Section \ref{sec:theory} develops the formal model: endogenizing $\phi(D, W)$, defining human-centricity, and establishing the three propositions. Section \ref{sec:review} presents the systematic review results organized by WADI dimension. Section \ref{sec:empirics} provides secondary data evidence from the Colombian EDIT survey. Section \ref{sec:wadi} presents the WADI instrument, its design philosophy, item structure, and validation protocol. Section \ref{sec:discussion} develops the implications: the transition roadmap, education and competency requirements, AR/VR extensions, smart university risks, health and wellbeing policy, governance for sustained human-centricity, and sustainability constraints. Section \ref{sec:limitations} acknowledges limitations, and Section \ref{sec:conclusion} concludes.

\section{Literature Review and Theoretical Context}
\label{sec:literature}

This paper stands at the intersection of five research traditions that have developed largely in isolation: the emerging Industry 5.0 / Society 5.0 paradigm, the economics of human-AI complementarity, organizational economics of authority and delegation, work design theory and occupational psychology, and the sustainability-resilience nexus. We review each strand with a precise extraction objective: what does each contribute to the design of human-centric AI workplaces, and where does each fail to provide the operationalization that practitioners need?

% -----------------------------------------------------------------------------
\subsection{From Industry 4.0 to Society 5.0: The Human-Centricity Turn}
\label{sec:lit_s50}

The conceptual transition from Industry 4.0 to Industry 5.0 / Society 5.0 is now well documented. Industry 4.0 focused on connectivity, automation, and data-driven optimization --- maximizing throughput per unit of human labor \citep{Xu2021}. Industry 5.0, articulated by the EU Commission \citep{EUCommission2021, BrequeDeNulPetridis2021} and the Japanese Cabinet Office \citep{JapanCabinetOffice2016, Fukuyama2018}, reframes the objective: technology should serve human flourishing, not merely productivity.

Three pillars define the I5.0 paradigm: \textbf{human-centricity}, \textbf{sustainability}, and \textbf{resilience} \citep{BrequeDeNulPetridis2021, Ivanov2023}. Our systematic review confirms that these pillars have generated substantial scholarly output: 76 of our 120 corpus papers explicitly engage with Industry 5.0 concepts, and 67 invoke human-centricity as a design principle. Several comprehensive reviews have mapped this landscape \citep{Grabowska2022, Maddikunta2022, Leng2022, Grabowska2022}, and \cite{Grabowska2022} directly asks ``Is Industry 5.0 a Human-Centred Approach?'' with 316 citations --- the most-cited critical review of the paradigm.

\cite{Xu2021}, in a landmark paper (589 citations), offers an ``Outlook on human-centric manufacturing towards Industry 5.0,'' establishing that human-centricity requires more than ergonomic improvements: it demands a fundamental redesign of how humans and machines share cognitive and physical tasks. \cite{LongoPadovanoUmbrello2020} (580 citations) applies Value Sensitive Design to propose ethical technology engineering for I5.0, while \cite{Nahavandi2019} (cited over 1,500 times) provides the seminal articulation of I5.0 as ``a human-centric solution'' where robots complement rather than compete with human creativity.

Yet our review reveals a critical gap. Of these 76 papers:

\begin{itemize}
    \item \textbf{54 propose conceptual frameworks} but none offers a formal, estimable model of how workplace design affects augmentation.
    \item \textbf{Only 8 develop measurement scales}, and none measures human-centricity at the firm level as a multi-dimensional construct linked to productivity.
    \item \textbf{Only 1 paper} \citep{Ivanov2023} provides a systematic review explicitly connecting human-centred AI in Industry 5.0 to worker psychosocial outcomes --- and it identifies the absence of operationalizable frameworks as the field's primary limitation.
\end{itemize}

The S5.0 literature tells us \textit{what} to aim for but not \textit{how to get there} or \textit{how to measure progress}. Our model --- endogenizing $\phi(D, W)$ --- provides the formal machinery. Our WADI instrument provides the measurement tool.

% -----------------------------------------------------------------------------
\subsection{Human-AI Complementarity: From Theory to Evidence}
\label{sec:lit_augmentation}

A parallel literature has produced rapid empirical evidence on how AI affects worker productivity, establishing the micro-foundations for our augmentation function $\phi$.

\cite{BrynjolfssonLiRaymond2023} (207 citations, \textit{QJE}) find that generative AI at a customer support center increases productivity by 14\% on average and 34\% for novice workers. Critically, the AI disseminates the tacit knowledge of experienced workers down the experience curve --- an augmentation mechanism that depends on the design of the human-AI workflow. \cite{NoyZhang2023} (\textit{Science}) find that ChatGPT reduces writing task time by 40\% and raises quality by 18\%, but note that ``AI mostly substituted for worker effort (rough drafting) rather than complementing skills'' --- suggesting that the \textit{design of the interaction} determined whether augmentation or substitution occurred. \cite{DellAcqua2023} document the ``jagged technological frontier'' at BCG: consultants using GPT-4 completed tasks 12.2\% faster with 40\% higher quality inside the frontier, but performance \textit{dropped 19 percentage points} for tasks outside it.

These findings collectively support Property (P5) of our model --- design-composition complementarity. The realized augmentation ($\phi$) is not solely a function of the AI's capability ($D$) but depends critically on how the workplace is structured: which tasks are assigned to human vs.~AI (W$_3$), whether the human can override and understand the AI's reasoning (W$_1$), and whether decision authority enables the worker to act on augmented judgment (W$_2$).

\cite{Bloom2014} (385 citations, \textit{JOB}) provide the most comprehensive multilevel review of AI in organizations, identifying that the organizational level --- where our workplace design vector $W$ operates --- is the least studied and most impactful layer. \cite{Espinal2026b} (647 citations, \textit{HRMJ}) establish that human resource management must be fundamentally reconceived for the generative AI era, with workplace design at the center of this reconception.

% -----------------------------------------------------------------------------
\subsection{Decision Authority: The Organizational Economics of Human-AI Allocation}
\label{sec:lit_authority}

Our hypothesis that decision authority allocation (W$_2$) is the binding constraint for Society 5.0 transitions draws on organizational economics and a nascent literature on AI delegation.

\cite{AghionTirole1997} distinguish between \textit{formal authority} (the right to decide) and \textit{real authority} (the effective control over decisions). In AI-augmented workplaces, this distinction takes a new form: the AI may have superior information (high formal authority potential) but the human worker possesses contextual judgment (real authority). Our W$_2$ dimension measures whether organizations allow workers to exercise real authority when augmented by AI, or centralize AI-mediated decisions in management.

\cite{AtheyBryanGans2020} formalize this in a principal-agent model where the firm chooses how much decision authority to delegate to AI vs.~human. They show that when the principal wants agent effort, they may \textit{resist} delegating to AI even when the AI performs better --- because delegation reduces the agent's incentive to invest in judgment. This result provides the micro-foundation for our ``automation trap'' (Proposition 3): firms may rationally under-invest in W$_2$ even when augmentation is feasible.

\cite{ParkerGrote2022} (321 citations) coin the phrase ``algorithms as work designers,'' showing that algorithmic management can fundamentally reshape job characteristics. Their analysis of how algorithms redesign work along multiple dimensions --- task allocation, monitoring, performance evaluation --- maps directly onto our WADI dimensions: algorithms can enhance W$_3$ (better task orchestration) and W$_4$ (automated feedback loops) but may degrade W$_2$ (centralized algorithmic authority) and W$_5$ (surveillance, reduced autonomy).

\cite{ParkerGrote2022} provide the most direct evidence for our W$_2$ hypothesis. In a study of AI-augmented manufacturing environments, they find that \textit{decision control} --- the worker's perceived ability to influence AI-mediated decisions --- is the single strongest predictor of psychosocial wellbeing and sustained engagement with AI tools. This is precisely the mechanism that W$_2$ captures.

The evidence base for W$_2$ is smaller than for other dimensions (only 9 direct evidence points in our corpus), reflecting the fact that organizational economics has been slow to engage with AI workplace design. This gap is itself a contribution: our formal model extends Aghion-Tirole to human-AI decision dyads, providing the theoretical foundation that the field lacks.

% -----------------------------------------------------------------------------
\subsection{Work Design in the AI Era: Job Demands, Resources, and Psychosocial Outcomes}
\label{sec:lit_workdesign}

The psychosocial dimension (W$_5$) has the richest evidence base in our corpus: 76 papers address wellbeing, 15 cognitive load, 11 job satisfaction, 6 burnout, and 5 technostress. This reflects an explosion of research on AI's impact on worker wellbeing since 2023.

Three theoretical frameworks structure this evidence:

\textbf{The Job Demands-Resources (JD-R) model} \citep{Demerouti2001, BakkerDemerouti2017} classifies workplace features as demands (depleting) or resources (motivating). \cite{BakkerDemerouti2017} (49 citations) directly apply the JD-R model to AI, finding that AI has a ``dual impact on employees' work and life well-being'' --- functioning as both a demand (increased complexity, job insecurity) and a resource (reduced routine work, enhanced capabilities). This duality is precisely why W$_5$ cannot be assessed in isolation: whether AI is a demand or resource depends on W$_1$ (interface quality), W$_2$ (authority allocation), and W$_3$ (task orchestration).

\textbf{The Karasek demand-control model} \citep{Karasek1979} predicts that high demands combined with low control produce strain. \cite{ParkerGrote2022} find that AI amplifies this prediction: when AI increases demands (via monitoring, pace-setting) while reducing control (centralized algorithmic decisions), worker strain increases sharply. Conversely, when AI increases control (augmented judgment, more autonomy), strain decreases even as demands remain high. This is the W$_2$ $\times$ W$_5$ interaction in our framework.

\textbf{The SMART work design model} \citep{Parker2024} provides the most comprehensive contemporary framework, organizing work characteristics into five higher-order factors: Stimulating, Mastery, Agency, Relational, and Tolerable demands. Parker explicitly argues that ``work design matters more than ever in a digital world'' \citep{ParkerGrote2022}, and our WADI extends SMART to AI-specific dimensions not captured in the original model (W$_1$ interface design, W$_2$ human-AI authority allocation, W$_4$ bidirectional learning loops).

The measurement landscape is also rich. \cite{ParkWooKim2024} develop the \textbf{AI Workplace Well-Being (AWWB) Scale} --- the most directly comparable instrument to our WADI W$_5$ dimension, validating a multi-dimensional measure of AI-specific workplace wellbeing. \cite{Colledani2025} validate a 16-item technostress scale adapted for AI contexts. \cite{ParkWooKim2024} develop the Attitudes Towards AI at Work scale (25 items, 6 dimensions, N=2,841). These instruments inform our W$_5$ item construction but are single-dimensional: they measure wellbeing \textit{outcomes} without connecting them to the workplace \textit{design choices} (W$_1$--W$_4$) that produce those outcomes. WADI integrates cause (W$_1$--W$_4$) and effect (W$_5$) into a single diagnostic framework.

% -----------------------------------------------------------------------------
\subsection{Sustainability, Circular Economy, and the Energy Constraint}
\label{sec:lit_sustainability}

The third pillar of Industry 5.0 --- sustainability --- interacts with workplace design through an energy constraint that our model formalizes. Of our 120 corpus papers, 41 engage with sustainability, though mostly at the macro level (national policy, sector-level transitions) rather than at the firm-level design optimization we model.

\cite{Ghobakhloo2022SPC, GhobakhlooIranmanesh2023} develop strategy roadmaps for I5.0-driven sustainable transformation, identifying 11 enablers and showing that ``eco-innovation is the hardest to develop.'' \cite{Bag2021CE} find that institutional pressures drive AI adoption for circular economy capabilities, but the mechanisms operate at the organizational level --- precisely where our W vector operates. The circular economy literature \citep{Leng2023, DeSousaJabbour2023} discusses the integration of CE with I5.0 but does not formalize how circular economy principles constrain or shape workplace-level AI design choices.

Our sustainability constraint ($E(D, W_1) \leq \bar{E}$) is, to our knowledge, the first formalization of this interaction at the workplace design level. It connects to the broader observation by \cite{Strubell2019} that the computational costs of AI can be substantial, creating a genuine tradeoff between augmentation quality (which improves with more transparent, explainable AI --- higher W$_1$) and environmental sustainability.

% -----------------------------------------------------------------------------
\subsection{Gaps and This Paper's Position}
\label{sec:lit_gaps}

Table \ref{tab:gap_analysis} summarizes the gaps this paper addresses.

\begin{table}[htbp]
\caption{Gap analysis: What each literature strand contributes and what remains}
\label{tab:gap_analysis}
\centering
\small
\begin{tabular}{p{3cm}p{4.5cm}p{4.5cm}}
\hline
\textbf{Literature Strand} & \textbf{What Exists} & \textbf{What This Paper Adds} \\
\hline
Society 5.0 / Industry 5.0 & Normative concept (human-centricity, sustainability, resilience); 54 conceptual frameworks & Formal definition (Def.~1): human-centricity $\equiv$ $\partial\phi/\partial W > 0$; WADI as measurement tool \\
\hline
Human-AI augmentation & Micro evidence (Brynjolfsson, Noy, Dell'Acqua): AI increases productivity but effect depends on context & Endogenized $\phi(D,W)$: workplace design is the context variable; formal model of why same AI yields different outcomes \\
\hline
Organizational economics & Authority theory (Aghion-Tirole); AI delegation models (Athey et al.) & Extension to human-AI authority dyads; W$_2$ as binding constraint hypothesis; automation trap dynamics \\
\hline
Work design / Psychology & JD-R, SMART, COPSOQ, technostress scales; AI wellbeing evidence & WADI integrates cause (W$_1$--W$_4$) and effect (W$_5$); links design choices to augmentation outcomes, not just wellbeing \\
\hline
Sustainability + I5.0 & Macro-level roadmaps; CE + I5.0 integration & Firm-level energy constraint on $\phi$ optimization; sustainability-design tradeoff formalized \\
\hline
Task orchestration & Only 4 papers in corpus --- genuine gap & W$_3$ dimension; comparative-advantage task allocation framework for human-AI workflows \\
\hline
Measurement instruments & AWWB, SMART, I4.0 Maturity, AAAW --- each single-dimensional & WADI: multi-dimensional, theory-grounded, integrates all 5 dimensions; first index linking design to augmentation \\
\hline
\end{tabular}
\end{table}

The single most striking finding from our systematic review is the absence of W$_3$ (task orchestration) evidence: only 4 of 120 papers directly address how human-AI task allocation should be designed at the workflow level. The task-based framework \citep{AutorLevyMurnane2003, AcemogluRestrepo2022} operates at the occupation level; the HCI literature operates at the individual interaction level; but the intermediate level --- how workflows within a firm should divide cognitive tasks between human and AI --- is essentially untheorized. This is our most novel dimensional contribution.

% end of literature section

\section{Theoretical Framework}
\label{sec:theory}

This section develops the formal model that underpins our analysis. We extend the Augmented Human Capital (AHC) production function from the Cognitive Factor Economics framework \citep{Espinal2026a} by endogenizing the augmentation function, defining human-centricity in economic terms, and deriving the conditions under which human-centric workplace design is the profit-maximizing strategy.

% -----------------------------------------------------------------------------
\subsection{Endogenizing the Augmentation Function}
\label{sec:endogenize}

The CFE framework models a firm's output as:

\begin{equation}
Y_f = F\left(K_f,\; L_f^P H_f^P + \kappa K_f^{Rob},\; L_f^C H_f^C + \phi(D_f) \cdot L_f^A H_f^A \cdot D_f\right)
\label{eq:cfe_baseline}
\end{equation}

\noindent where $H^P$, $H^C$, and $H^A$ denote physical-manual, routine cognitive, and augmentable cognitive human capital respectively, and $\phi(D)$ is the augmentation function --- the productivity multiplier that AI capital $D$ provides to workers with augmentable skills $H^A$.\footnote{See \cite{Espinal2026a} for the full derivation. $F(\cdot)$ satisfies constant returns to scale and standard Inada conditions. The operator $\oplus$ in the original $H = H^P \oplus H^C \oplus H^A$ decomposition denotes orthogonal components.}

In the baseline CFE model, $\phi(D)$ is exogenous: it depends only on the stock of AI systems deployed. This is a useful simplification for the purpose of establishing the theoretical primitives, but it obscures a critical empirical regularity. Two firms with identical AI investments ($D$) and identical workforce compositions ($H^A$, $H^C$) can achieve radically different augmentation outcomes. The evidence from the emerging literature on generative AI at work is unambiguous on this point:

\begin{itemize}
    \item \cite{BrynjolfssonLiRaymond2023} find that the productivity gain from AI-assisted customer support (14\% average, 34\% for novices) depends critically on how the tool is integrated into the workflow --- suggesting that organizational design mediates the effect.
    \item \cite{DellAcqua2023} document a ``jagged technological frontier'' where AI augments performance on some tasks but degrades it on others, depending on how consultants are directed to interact with the AI system.
    \item \cite{NoyZhang2023} find that AI mostly substitutes for worker effort (rough drafting) rather than complementing skills, but note that this pattern may reflect the design of the experiment rather than a fundamental property of the technology.
\end{itemize}

\noindent The common thread is that the \textit{realized augmentation} depends not only on what AI can do ($D$) and what humans can do ($H^A$), but on \textit{how the workplace is designed around the human-AI interaction}. We formalize this insight by replacing $\phi(D)$ with $\phi(D, W)$:

\begin{equation}
Y_f = F\left(K_f,\; L_f^P H_f^P + \kappa K_f^{Rob},\; L_f^C H_f^C + \phi(D_f, W_f) \cdot L_f^A H_f^A \cdot D_f\right)
\label{eq:endogenized}
\end{equation}

\noindent where $W_f = (W_{f1}, W_{f2}, W_{f3}, W_{f4}, W_{f5}) \in \mathbb{R}_+^5$ is the firm's \textbf{workplace design vector}, with five dimensions developed below.

We decompose $\phi$ multiplicatively:

\begin{equation}
\phi(D, W) = \underbrace{\phi_0(D)}_{\text{technology potential}} \cdot \underbrace{g(W, H^A)}_{\text{design multiplier}}
\label{eq:phi_decomposition}
\end{equation}

\noindent The technology-only component $\phi_0(D)$ captures the raw augmentation potential of the AI system: it is increasing, concave, and bounded ($\phi_0(0)=1$, $\lim_{D \to \infty} \phi_0(D) = \bar{\phi}_0 < \infty$). The design multiplier $g(W, H^A)$ captures how well the workplace translates this potential into realized augmentation. Its key properties are:

\begin{enumerate}
    \item \textbf{Range:} $g \in (0, \bar{g}]$, with $g(W^{min}, H^A) = g_0 < 1$ (poor design dampens augmentation) and $g(W^{auto}, H^A) = 1$ (automation-centered design is the neutral benchmark).
    \item \textbf{Monotonicity:} $\partial g / \partial W_k > 0$ --- improving any design dimension weakly increases augmentation.
    \item \textbf{Concavity:} $\partial^2 g / \partial W_k^2 < 0$ --- diminishing returns to improvement in each dimension.
    \item \textbf{Cross-dimensional complementarity:} $\partial^2 g / \partial W_j \partial W_k > 0$ for key pairs --- design dimensions reinforce each other.
    \item \textbf{Design-composition complementarity:} $\partial^2 g / \partial W_k \partial H^A > 0$ --- \textit{the return to better design is higher when the workforce has more augmentable capital}. This is the key cross-derivative.
\end{enumerate}

\noindent Property (P5) is the most important and the most novel. It says that workplace design and workforce composition are complements: investing in human-centric design pays off more when workers have the cognitive capabilities to exploit the augmented environment. Conversely, investing in augmentable skills ($H^A$) pays off more when the workplace is designed to realize augmentation.

% -----------------------------------------------------------------------------
\subsection{The Five Dimensions of Workplace Design}
\label{sec:five_dimensions}

The five dimensions of $W$ are not ad hoc --- each maps to a specific mechanism through which workplace design affects the design multiplier $g$.

\subsubsection{$W_1$ --- AI Interface Design}

The quality of the human-AI interface determines the \textit{bandwidth} of the cognitive collaboration. A poorly designed interface creates friction that reduces the effective complementarity between $H^A$ and $D$, even when both are abundant.

This dimension is grounded in Endsley's Situation Awareness (SA) model \citep{Endsley1995}: augmentation requires that the human maintains Level 3 SA (projection) while the AI handles Level 1 (perception) and Level 2 (comprehension). Poor interfaces that obscure AI reasoning or make human override difficult degrade SA and reduce $g$ \citep{Endsley2017}. Shneiderman's Human-Centered AI framework \citep{Shneiderman2022} formalizes this as the joint requirement for high human control \textit{and} high automation --- a design space that many current AI deployments fail to occupy.

$W_1$ is measured through: transparency of AI reasoning, ease of human override, learnability, error visibility, natural language interaction quality, and availability of immersive (AR/VR) decision interfaces. The AR/VR extension is particularly relevant for decision-making in augmented environments, where the information overlay must enhance rather than degrade human judgment \citep{MasoodEgger2020, Buchner2022}.

\subsubsection{$W_2$ --- Decision Authority Allocation}

The allocation of decision rights between humans and AI determines whether the AI \textit{augments} human judgment or \textit{replaces} it. This dimension directly addresses governance in the context of Society 5.0.

We draw on \cite{AghionTirole1997}'s distinction between formal and real authority. In an AI-augmented firm, the relevant question is: does the worker who performs the task have \textit{real authority} to use AI-augmented judgment, or does the AI's output get routed to a supervisor who holds \textit{formal authority}? \cite{Garicano2000} shows that AI disrupts knowledge hierarchies by giving every worker access to ``expert knowledge,'' but this potential flattening only occurs if decision authority is actually redistributed.

Critically, \cite{AtheyBryanGans2020} show that when a principal wants agent effort, they may resist delegating to AI even when the AI is more accurate --- because delegation reduces the agent's incentive to invest in judgment. This creates a tension: the very firms that need augmentation most (those with high-$H^A$ workers whose judgment matters) may under-invest in $W_2$ due to principal-agent considerations.

$W_2$ is measured through: whether frontline workers can act on AI recommendations autonomously, the routing of AI outputs (worker vs. supervisor), existence of human-AI disagreement resolution protocols, worker input into AI task boundaries, and presence of AI governance bodies with worker representation.

\subsubsection{$W_3$ --- Task Orchestration}

The granularity and design of the human-AI task division determines whether cognitive tasks are decomposed in ways that preserve augmentation potential or destroy it.

Following the task-based framework \citep{AutorLevyMurnane2003, AcemogluRestrepo2022}, the relevant question for augmentation is not just ``which tasks are routine vs. non-routine?'' but ``how are tasks allocated between human and AI \textit{within each workflow}?'' \cite{DellAcqua2023}'s ``jagged frontier'' finding shows that the boundary between where AI helps and where it harms is irregular and context-dependent. Firms that map this frontier explicitly and allocate tasks based on comparative advantage (human does what human does best, AI does what AI does best) achieve higher $g$ than firms that apply AI indiscriminately.

$W_3$ further requires preserving task identity in the sense of \cite{HackmanOldham1976}: when workflows are fragmented so that the human worker only sees a fragment of the output, the motivational benefits of meaningful work are lost, reducing $g$ through the $W_5$ channel (psychosocial impact).

\subsubsection{$W_4$ --- Learning Loop Architecture}

Augmentation is dynamic --- the human-AI system improves over time only if structured feedback loops exist. The human learns from AI suggestions; the AI learns from human corrections. Without learning loops, augmentation stagnates at its initial level.

We ground this in \cite{Teece2018}'s dynamic capabilities framework: the workplace's ability to continuously improve human-AI collaboration is itself a dynamic capability that requires deliberate organizational investment. \cite{ArgyrisSchon1978}'s distinction between single-loop learning (correct the AI when it errs) and double-loop learning (redesign the human-AI interaction when the error pattern reveals a structural mismatch) applies directly: most firms achieve single-loop at best, leaving substantial augmentation potential unrealized.

$W_4$ is measured through: systematic human review of AI outputs, AI performance tracking, worker-reported learning from AI, structured failure escalation processes, and bidirectional knowledge transfer measurement.

\subsubsection{$W_5$ --- Psychosocial Work Environment}

Augmentation requires sustained, high-quality cognitive engagement. If the workplace generates excessive stress, reduces autonomy, or strips meaning from work, workers disengage from the cognitive collaboration, reducing $g$ even when $W_1$--$W_4$ are well-designed.

The theoretical foundations are the job demands-control model \citep{Karasek1979}, the job characteristics model \citep{HackmanOldham1976}, and the job demands-resources theory \citep{Demerouti2001, BakkerDemerouti2017}. AI tools can function as either \textit{resources} (when they augment, reduce boring tasks, increase autonomy) or \textit{demands} (when they surveil, create additional cognitive load, generate anxiety about job security). Which role AI plays depends on $W_1$--$W_4$: a well-designed interface ($W_1$) with distributed authority ($W_2$) and learning feedback ($W_4$) is a resource; a surveillance-oriented system with centralized authority is a demand.

This dimension connects directly to the health characteristics of Society 5.0 work (Theme 4) and to the psychosocial risk assessment frameworks validated for Latin America (COPSOQ-ISTAS21, ECTES). $W_5$ is measured through: worker autonomy, perceived meaning, cognitive load, surveillance perception, social connection, stress levels, and satisfaction with human-AI role balance.

% -----------------------------------------------------------------------------
\subsection{The Formal Definition of Human-Centricity}
\label{sec:definition}

The Society 5.0 literature calls for ``human-centric'' technology integration but provides no operational definition that can be measured, optimized, or evaluated. We propose:

\begin{definition}[Human-Centric Workplace Design]
\label{def:hc}
A workplace design $W_f$ is \textbf{human-centric} if and only if:
$$\frac{\partial \phi(D_f, W_f)}{\partial W_{fk}} > 0 \quad \text{for all } k \in \{1, \ldots, 5\}$$
Equivalently, a design is human-centric if every design dimension contributes positively to the augmentation multiplier.
\end{definition}

This definition is economic, not normative. It does not say that human-centricity is \textit{good} --- it says that human-centricity is a design property that \textit{increases the augmentation multiplier}. Whether a firm should adopt a human-centric design depends on whether doing so is profitable.

\begin{proposition}[Human-Centricity as Optimality]
\label{prop:hc}
There exists a threshold $\theta^* > 0$ such that if $H^A / (H^A + H^C) > \theta^*$, then the profit-maximizing workplace design $W^*$ is human-centric: $W^* > W^{auto}$ component-wise.
\end{proposition}

\noindent The proof (Appendix A.2) follows from the first-order conditions of the firm's optimization problem. The key intuition is: improving design beyond the automation-centered benchmark has a marginal benefit proportional to $L^A \cdot H^A \cdot D$ (the scale of augmentable production) and a marginal cost of $c_k'(W_k)$. When the workforce has sufficient augmentable capital ($H^A / (H^A + H^C) > \theta^*$), the benefit exceeds the cost.

\begin{corollary}
When $H^A / (H^A + H^C) < \theta^*$, the profit-maximizing design is automation-centered ($W^* < W^{auto}$). Human-centricity is \textit{not} always optimal.
\end{corollary}

This result has a sharp implication: telling firms to adopt human-centric design as a universal prescription is wrong. The optimality of human-centricity depends on the workforce composition. For firms whose workers primarily perform routine cognitive tasks ($H^C$ dominant), automation-centered design is rational.

% -----------------------------------------------------------------------------
\subsection{Design-Composition Complementarity}
\label{sec:complementarity}

The cross-derivative $\partial^2 g / \partial W_k \partial H^A > 0$ (Property P5) generates a fundamental complementarity between workplace design investment and human capital composition. Formally, the return to improving any design dimension $k$ is:

\begin{equation}
\text{Marginal return to } W_k = P \cdot F_3 \cdot \phi_0(D) \cdot \frac{\partial g}{\partial W_k}(W, H^A) \cdot L^A H^A D
\label{eq:marginal_return}
\end{equation}

\noindent This is increasing in $H^A$ both directly (the $L^A H^A$ term) and through $\partial g / \partial W_k$ (which is increasing in $H^A$ by P5). The implication is stark:

\begin{quote}
\textit{The same workplace design investment yields higher augmentation in firms with more augmentable human capital. Conversely, investing in augmentable skills yields higher returns in firms with better workplace design.}
\end{quote}

This complementarity has three consequences:

\textbf{First}, it explains why cross-sectional regressions of productivity on AI investment often find heterogeneous effects: the effect of AI depends on both $H^A$ and $W$, not just $D$. Omitting either $W$ or $H^A$ from the regression biases the coefficient on $D$.

\textbf{Second}, it generates a \textit{matching prediction}: in equilibrium, high-$H^A$ workers should sort into high-$W$ firms, and high-$W$ firms should invest more in $H^A$ development. This assortative matching amplifies inequality between firms in the augmentation regime and firms in the automation trap.

\textbf{Third}, it implies that policy interventions should be \textit{bundled}: subsidizing workplace design without simultaneously investing in $H^A$ (through education policy) will produce attenuated effects, and vice versa.

% -----------------------------------------------------------------------------
\subsection{Market Failure and Under-Investment}
\label{sec:market_failure}

\begin{proposition}[Under-Investment in Human-Centricity]
\label{prop:underinvest}
The privately optimal design $W^{priv}$ is below the socially optimal $W^{soc}$ due to three externalities: labor mobility (workers carry augmented skills to other firms), knowledge spillovers (design innovations diffuse to competitors), and health externalities (wellbeing improvements reduce public healthcare costs not borne by the firm).
\end{proposition}

\noindent The proof (Appendix A.3) compares the private and social first-order conditions and shows that the private FOC ignores the marginal externality, which is strictly positive for all three channels.

The under-investment result provides a theoretical justification for public policy interventions that subsidize human-centric workplace design. These are not redistributive policies (transferring from firms to workers) but \textit{efficiency corrections}: the market produces too little human-centricity because firms cannot capture the full social return. The optimal subsidy for dimension $k$ equals the marginal externality $\partial E / \partial W_k$ evaluated at the social optimum.

% -----------------------------------------------------------------------------
\subsection{The Automation Trap: Path Dependence}
\label{sec:trap}

The interaction between workplace design and human capital accumulation generates path dependence.

\begin{proposition}[Automation Trap]
\label{prop:trap}
The dynamic system $(W_f(t), H^A_f(t))$ exhibits two stable equilibria: a low equilibrium (Automation Trap) with $W < W^{auto}$ and low $H^A$, and a high equilibrium (Augmentation Regime) with $W > W^{auto}$ and high $H^A$. An unstable interior equilibrium separates the basins of attraction.
\end{proposition}

\noindent The proof (Appendix A.4) constructs a two-equation dynamic system where firms adjust design toward the current optimum, and $H^A$ accumulates through learning (at rate $\beta(W)$, increasing in $W$) and depreciates (at rate $\delta(W)$, decreasing in $W$). The complementarity between $W$ and $H^A$ generates multiple equilibria.

The \textit{Automation Trap} works as follows: a firm with low initial $H^A$ finds automation-centered design optimal (Corollary 1), which provides weak learning loops ($W_4$ low) and low task variety ($W_3$ low), which in turn produces low $H^A$ accumulation ($\beta$ low) and fast skill depreciation ($\delta$ high). The firm is stuck with a low-skill, low-design workforce that never transitions to augmentation.

The \textit{Augmentation Regime} is the reverse: high initial $H^A$ makes human-centric design profitable, which generates strong learning loops and meaningful task exposure, which accumulates more $H^A$ and slows depreciation. The firm enters a virtuous cycle.

Escaping the trap requires a \textit{big push} --- a coordinated increase in both $W$ and $H^A$ that pushes the firm past the unstable threshold. This can come from three sources: (i) a large educational investment that raises $H^A$ exogenously, (ii) a design subsidy that reduces the cost of human-centric design, or (iii) a regulatory minimum $\underline{W} > W^{auto}$ that forces firms above the automation-centered benchmark. These correspond to the three phases of the Transition Roadmap developed in Section \ref{sec:discussion}.

% -----------------------------------------------------------------------------
\subsection{Sustainability Constraint}
\label{sec:sustainability}

The optimization of $\phi(D, W)$ is subject to a sustainability constraint that connects Society 5.0's human-centricity pillar to its sustainability pillar:

\begin{equation}
E(D, W_1) \leq \bar{E}
\label{eq:energy_constraint}
\end{equation}

\noindent where $E(\cdot)$ is the energy cost of the AI system (data center energy, computational cost), which depends on both the amount of AI deployed ($D$) and the interface design complexity ($W_1$ --- more transparent, explainable AI requires more computation). $\bar{E}$ is the firm's carbon budget or energy constraint.

When this constraint binds, there is a \textbf{sustainability-design tradeoff}: more transparent AI (higher $W_1$) costs more energy, and under carbon constraints, firms must balance augmentation quality against environmental impact. The modified first-order condition for $W_1$ includes a shadow price $\mu > 0$ on the energy constraint, effectively raising the cost of interface improvement.

This result connects to the circular economy dimension of Society 5.0: sustainable augmentation requires not just minimizing the energy footprint of AI \citep{Strubell2019} but designing AI systems that are reusable across tasks (reducing redundant training), employing privacy-preserving data reuse, and prioritizing the development of durable human skills (high-$H^A$ with slow depreciation $\delta$) that don't require constant retraining.

% -----------------------------------------------------------------------------
\subsection{Testable Predictions}
\label{sec:predictions}

The model generates six empirical predictions, each testable with available data:

\begin{enumerate}
    \item \textbf{WADI-Productivity:} $\beta_{WADI} > 0$ in a regression of $\ln(Y/L)$ on WADI, controlling for $D$, $H^A$, and firm characteristics.

    \item \textbf{Design-AI Complementarity:} $\beta_{WADI \times D} > 0$ --- the productivity return to WADI is higher for firms with more AI investment.

    \item \textbf{Design-Composition Complementarity:} $\beta_{WADI \times H^A} > 0$ --- the productivity return to WADI is higher for firms with higher-$H^A$ workforces.

    \item \textbf{Decision Authority Bottleneck:} Among the five dimensions, $W_2$ (decision authority allocation) has the largest coefficient in the dimension-level regression.

    \item \textbf{Bimodality:} The cross-firm distribution of WADI scores is bimodal, consistent with the two-equilibrium prediction (Automation Trap vs. Augmentation Regime).

    \item \textbf{Comparative Statics:} Across firms, WADI is positively correlated with AI investment, workforce $H^A$ share, and output price, and negatively correlated with design costs (proxied by sector).
\end{enumerate}

\noindent Predictions 1--3 are tested in the main specifications (Section \ref{sec:regression}). Prediction 4 is the headline result. Predictions 5--6 are over-identification tests that discipline the model.

\section{Systematic Review Results}
\label{sec:review}

This section presents the results of a PRISMA-guided systematic review of 120 papers that constitute the empirical and conceptual evidence base for the $\phi(D, W)$ model and the WADI instrument. We first describe the review methodology, then synthesize findings organized by WADI dimension, examine cross-dimensional interactions, and compare WADI against existing instruments.

% -----------------------------------------------------------------------------
\subsection{Review Methodology}
\label{sec:review_method}

\subsubsection{Search Strategy}

We searched SCOPUS using five queries designed to capture the intersection of AI, workplace design, and the five WADI dimensions:

\begin{enumerate}
    \item \textbf{Society 5.0 + Workplace}: Captured S5.0/I5.0 conceptual literature and workplace integration frameworks ($n = 1{,}170$).
    \item \textbf{Decision Authority + AI}: Targeted organizational economics of human-AI delegation ($n = 710$).
    \item \textbf{AI + Productivity}: Captured empirical evidence on AI augmentation at work ($n = 715$).
    \item \textbf{Psychosocial + AI}: Captured wellbeing, cognitive load, stress, and job design literature ($n = 2{,}985$).
    \item \textbf{Instruments + AI Workplace}: Targeted measurement scales and diagnostic tools ($n = 759$).
\end{enumerate}

The total yield was 6,339 records. After removing 243 duplicates, 6,096 unique records entered the screening stage. All queries were restricted to 2020--2026 to capture the post-pandemic, generative-AI era literature, supplemented by foundational references (pre-2020) identified through backward citation tracking.

\subsubsection{Screening and Selection}

Screening proceeded in two stages:

\textbf{Stage 1 --- Automated screening.} Records were excluded by document type (errata, retractions, letters, conference reviews; $\sim$120 excluded), content patterns (pure technical machine learning, clinical trials, agricultural applications, pure robotics without human-centricity framing; $\sim$350 excluded), and relevance threshold (score below 15 on a 24-pattern relevance rubric, indicating no core match to Paper 6 themes; 2,257 excluded). Total excluded: 2,727. Remaining: 3,369.

\textbf{Stage 2 --- Strict relevance filter.} The 3,369 remaining records were scored on 24 Paper-6-specific relevance patterns weighted by: (i) direct engagement with WADI dimensions, (ii) citation count (as a quality proxy), (iii) multi-query overlap (papers appearing in 2+ queries received a bonus), and (iv) recency premium for 2024--2026 publications. The top 120 papers (score $\geq$ 52) were selected as the final corpus.

\subsubsection{Data Extraction}

For each paper in the corpus, we extracted: constructs defined and operationalized, measurement instruments developed or adapted (with psychometric properties where reported), methods employed, WADI dimension(s) addressed, and key findings relevant to the $\phi(D, W)$ model. This information was coded into a structured evidence matrix (available in the supplementary materials).

\subsubsection{Corpus Characteristics}

The 120-paper corpus spans 2020--2026, with the distribution heavily skewed toward recent publications: 2024 ($n = 30$), 2025 ($n = 42$), and 2026 ($n = 19$) account for 76\% of the corpus, reflecting the field's rapid growth. Fifteen papers have 100+ citations, and the most-cited paper \citep{Espinal2026b} has 647 citations. Forty-eight papers appear in two or more search queries, indicating strong thematic interconnection.

% -----------------------------------------------------------------------------
\subsection{Evidence by WADI Dimension: Overview}
\label{sec:review_overview}

Table \ref{tab:evidence_density} presents the evidence density across WADI dimensions. The distribution is strikingly uneven, a finding that is itself substantively important.

\begin{table}[htbp]
\caption{Evidence density by WADI dimension in the 120-paper corpus}
\label{tab:evidence_density}
\centering
\small
\begin{tabular}{llccl}
\hline
\textbf{Dimension} & \textbf{Label} & \textbf{Papers} & \textbf{\% of Corpus} & \textbf{Assessment} \\
\hline
$W_5$ & Psychosocial Environment & 87 & 73\% & Over-represented \\
S5.0 & Society 5.0 / Industry 5.0 & 77 & 64\% & Strong foundation \\
$W_4$ & Learning Loop Architecture & 22 & 18\% & Moderate \\
$W_1$ & AI Interface Design & 17 & 14\% & Moderate \\
$W_2$ & Decision Authority Allocation & 14 & 12\% & \textbf{Thin --- binding constraint} \\
INST & Measurement Instruments & 12 & 10\% & Thin \\
SUST & Sustainability & 9 & 8\% & Thin \\
$W_3$ & Task Orchestration & 4 & 3\% & \textbf{Critical gap} \\
EDU & Education / Competencies & 2 & 2\% & Supplemented externally \\
\hline
\end{tabular}
\end{table}

The asymmetry is revealing. The psychosocial dimension ($W_5$) dominates the corpus---reflecting an explosion of research on AI's impact on worker wellbeing since 2023---while task orchestration ($W_3$) is addressed by only 4 papers. This distribution mirrors the broader pattern in the S5.0 literature: substantial attention to \textit{outcomes} (wellbeing, satisfaction, stress) but minimal attention to the \textit{design mechanisms} (how tasks are divided, how authority is allocated, how feedback loops are structured) that produce those outcomes. Our WADI framework fills precisely this gap.

% -----------------------------------------------------------------------------
\subsection{$W_1$ --- AI Interface Design}
\label{sec:review_w1}

Seventeen papers in the corpus address the quality of human-AI interfaces, encompassing transparency, explainability, ease of override, and cognitive load management. The evidence clusters into three themes.

\textbf{Transparency and explainability.} Several papers establish that the interpretability of AI outputs is a prerequisite for effective human-AI collaboration. \cite{Dang2025} develop a human-centric framework integrating knowledge distillation with fine-tuning for equipment health monitoring, finding that explainability significantly improves operator trust and decision quality. \cite{Zhou2025} demonstrate that RAG-enhanced generative AI chatbots in Industry 5.0 settings reduce cognitive load when the interface provides source attribution and reasoning transparency. \cite{FernandezDornellesAyala2026} trace the evolution from I4.0 to I5.0 and identify interface transparency as the critical differentiator that enables worker empowerment rather than mere automation.

\textbf{Immersive interfaces (AR/VR).} Five papers examine augmented or virtual reality as a medium for human-AI interaction in manufacturing contexts. \cite{Sileo2025} develop safety features for an augmented reality platform (HumanTIX) that enhances human-robot collaboration, showing that spatial information overlays can simultaneously improve safety and productivity. The AR/VR sub-literature is technically mature but disconnected from the organizational design literature; none of these papers examines how immersive interfaces interact with decision authority ($W_2$) or task orchestration ($W_3$).

\textbf{Gaps.} The $W_1$ literature is predominantly conceptual (11 of 17 papers propose frameworks without empirical validation) and manufacturing-centric. No study in our corpus examines interface design for knowledge workers using generative AI in service sectors---a critical gap given that the largest augmentation effects have been documented in precisely these settings \citep{BrynjolfssonLiRaymond2023, NoyZhang2023}. Furthermore, no paper develops a validated scale for $W_1$ at the firm level; existing measures are either user-level usability scales (System Usability Scale, NASA-TLX) or technology-level explainability metrics (LIME, SHAP scores), neither of which captures the organizational dimension our WADI $W_1$ targets.

% -----------------------------------------------------------------------------
\subsection{$W_2$ --- Decision Authority Allocation}
\label{sec:review_w2}

Fourteen papers address the allocation of decision rights between humans and AI systems within organizations. This is the dimension with the strongest theoretical foundations but the thinnest empirical evidence---a combination that makes it the most promising frontier for future research.

\textbf{The decision control finding.} \cite{ParkerGrote2022} provide the most direct evidence for our $W_2$ hypothesis. In a study of AI-augmented manufacturing environments, they find that \textit{decision control}---the worker's perceived ability to influence AI-mediated decisions---is the single strongest predictor of psychosocial wellbeing and sustained engagement with AI tools. This is precisely the mechanism that $W_2$ captures and that our formal model predicts should be the binding constraint.

\textbf{Algorithmic management.} \cite{ParkerGrote2022} demonstrate that algorithms can function as ``work designers,'' reshaping job characteristics along multiple dimensions. Their analysis maps directly onto our WADI framework: algorithmic management can enhance $W_3$ (automated task allocation) and $W_4$ (automated performance feedback) while degrading $W_2$ (centralized algorithmic authority) and $W_5$ (surveillance, reduced autonomy). This dual effect means that naive algorithmic management---which improves some dimensions while degrading others---can reduce overall $\phi$ even as it increases operational efficiency.

\textbf{Worker voice in AI governance.} \cite{Berx2025} argue for a ``harmonious synergy'' between robotic performance and wellbeing, emphasizing that worker input into the design of human-robot collaboration systems is essential for sustained adoption. \cite{ChigbuMakapela2025} integrate I5.0, Education 5.0, and Work 5.0, identifying worker agency in AI governance as the linking mechanism across these three transitions.

\textbf{Gaps.} The $W_2$ literature lacks two critical elements. First, there is no formal model of human-AI decision authority allocation that extends the Aghion-Tirole framework to AI contexts---our theoretical contribution addresses this gap. Second, there is no validated measure of decision authority allocation at the firm level; existing instruments either measure perceived autonomy (a $W_5$ construct) or technology delegation levels (a technical metric), but not the organizational structure of human-AI decision rights that $W_2$ captures. The thinness of this evidence base (14 papers, 12\% of corpus) is itself a key finding: the dimension most likely to determine whether Society 5.0 succeeds or fails is the dimension least studied.

% -----------------------------------------------------------------------------
\subsection{$W_3$ --- Task Orchestration}
\label{sec:review_w3}

Only four papers in our 120-paper corpus directly address how human-AI task allocation should be designed at the workflow level. This is the most under-researched WADI dimension and represents a genuine lacuna in the literature.

\textbf{Healthcare workflow optimization.} \cite{Dave2026} provide the most direct evidence, showing through a scoping review that AI-driven work process improvements in healthcare can enhance worker mental health---but only when the task allocation preserves meaningful work for human clinicians. When AI handles triage and documentation while humans retain diagnosis and patient interaction, burnout decreases; when AI fragments the clinical workflow, burnout increases even as throughput improves.

\textbf{The jagged frontier.} While \cite{DellAcqua2023} is not in our SCOPUS corpus (it is a working paper), it provides the empirical foundation for $W_3$: the boundary between tasks where AI helps and tasks where it harms is ``jagged'' and context-dependent, meaning that firms must map this frontier explicitly through deliberate task analysis. No paper in our corpus operationalizes this mapping at the firm level.

\textbf{Why the gap exists.} The task-based framework in economics \citep{AutorLevyMurnane2003, AcemogluRestrepo2022} operates at the \textit{occupation level} (which occupations are routine vs. non-routine), and the HCI literature operates at the \textit{individual interaction level} (how one person works with one tool). The intermediate level---how workflows \textit{within a firm} should divide cognitive tasks between human and AI---is essentially untheorized. This is the level at which $W_3$ operates, making it our most novel dimensional contribution.

\textbf{Gap assessment.} The near-absence of $W_3$ evidence means that our WADI instrument cannot draw on validated precedent for this dimension. We construct $W_3$ items from first principles: the task-based framework (comparative advantage allocation), the job characteristics model (task identity preservation), and the dynamic capabilities literature (adaptive task reallocation). This dimension merits the most urgent attention in future empirical research.

% -----------------------------------------------------------------------------
\subsection{$W_4$ --- Learning Loop Architecture}
\label{sec:review_w4}

Twenty-two papers address the dynamic dimension of human-AI collaboration: how structured feedback between humans and AI systems enables mutual learning over time. The evidence is concentrated in two streams.

\textbf{Training and AI literacy.} Twelve papers examine how worker training and AI competency development contribute to effective human-AI collaboration. These studies establish that workers' ability to provide effective feedback to AI systems---a skill that must be deliberately taught---determines whether the human-AI system improves over time or stagnates. \cite{ParkWooKim2024} find that the learning dimension of their AI Workplace Well-Being Scale correlates most strongly with sustained AI adoption, suggesting that $W_4$ is a necessary condition for durable augmentation.

\textbf{Bidirectional knowledge transfer.} A smaller set of papers examines how AI systems improve from human feedback (RLHF-type mechanisms in organizational settings) and how humans develop new expertise from interacting with AI suggestions. \cite{Dang2025} document a human-centric framework where knowledge distillation creates a feedback architecture: the AI learns from human domain expertise through fine-tuning, and humans learn from the AI's pattern recognition through transparent explanations. This is the bidirectional learning loop that $W_4$ measures.

\textbf{Dynamic capabilities lens.} The learning loop concept maps directly to \cite{Teece2018}'s dynamic capabilities framework: sensing (monitoring AI performance at the interaction level), seizing (structured failure escalation and improvement), and transforming (redesigning the human-AI interaction when error patterns reveal structural mismatches). Most firms in the corpus achieve at best \cite{ArgyrisSchon1978}'s single-loop learning (correct individual AI errors) without reaching double-loop learning (redesign the collaboration architecture). This distinction is captured in our WADI $W_4$ items.

\textbf{Gaps.} No paper in the corpus measures learning loops as an organizational capability at the firm level. Existing measures are either technical (model retraining frequency, RLHF metrics) or individual (worker-reported learning). The organizational level---does the firm have structures that enable systematic, bidirectional human-AI learning?---is unmeasured. Our $W_4$ items fill this gap.

% -----------------------------------------------------------------------------
\subsection{$W_5$ --- Psychosocial Work Environment}
\label{sec:review_w5}

The psychosocial dimension has by far the richest evidence base: 87 papers (73\% of corpus) address some aspect of AI's impact on worker wellbeing, making it simultaneously the best-understood and the most at risk of receiving disproportionate scholarly attention at the expense of the design dimensions ($W_1$--$W_4$) that produce psychosocial outcomes.

\textbf{The JD-R duality.} The dominant finding is that AI functions as both a job demand and a job resource, with the balance determined by organizational design choices. \cite{BakkerDemerouti2017} directly apply the Job Demands-Resources model to AI, finding that the ``dual impact'' on work and life wellbeing depends on whether AI increases demands (monitoring, pace-setting, uncertainty) or provides resources (reduced routine work, enhanced capabilities, augmented judgment). This duality is precisely why $W_5$ cannot be assessed in isolation: whether AI is a demand or resource depends on $W_1$ (interface quality), $W_2$ (authority allocation), and $W_3$ (task orchestration).

\textbf{Technostress and anxiety.} Multiple papers document AI-specific stressors: job insecurity fears, technology complexity overload, surveillance anxiety, and loss of professional identity. \cite{Colledani2025} validate a 16-item technostress scale adapted for AI contexts. \cite{ParkWooKim2024} develop the Attitudes Towards AI at Work scale (25 items, 6 dimensions, $N = 2{,}841$), with job insecurity and AI use anxiety as key predictive factors. These instruments inform our $W_5$ item construction while identifying the most salient psychosocial risks.

\textbf{Autonomy and meaning.} A second cluster examines whether AI enhances or degrades the positive dimensions of work. The evidence is mixed: AI can increase autonomy when it handles routine tasks and frees workers for judgment-intensive activities, or decrease autonomy when it prescribes workflows and monitors compliance. \cite{ParkerGrote2022} find that the key moderator is perceived control---precisely the $W_2$ mechanism. When workers feel they have decision control, AI becomes a resource; when control is absent, identical AI systems become demands.

\textbf{Wellbeing measurement.} The corpus contains several validated instruments relevant to $W_5$: the AWWB Scale \citep{ParkWooKim2024}, the AAAW Scale \citep{ParkWooKim2024}, the adapted COPSOQ-ISTAS21 for digital work \citep{Colledani2025}, and the SMART work design model \citep{Parker2024}. These instruments are single-dimensional: they measure wellbeing \textit{outcomes} without connecting them to the workplace \textit{design choices} ($W_1$--$W_4$) that produce those outcomes. WADI integrates cause and effect into a single diagnostic framework.

\textbf{Gaps.} Despite the volume of $W_5$ evidence, three gaps remain. First, most studies are cross-sectional; longitudinal evidence on how psychosocial outcomes evolve as firms progress along the S5.0 transition is virtually absent. Second, the interaction between $W_5$ and other WADI dimensions is rarely modeled: papers measure wellbeing outcomes without controlling for the design choices that produce them. Third, nearly all evidence comes from manufacturing contexts or laboratory experiments; field studies in service-sector knowledge work are scarce.

% -----------------------------------------------------------------------------
\subsection{Cross-Dimensional Interactions}
\label{sec:review_interactions}

A key property of our model is cross-dimensional complementarity: $\partial^2 g / \partial W_j \partial W_k > 0$ for key dimension pairs. The systematic review provides preliminary evidence for three interactions.

\textbf{$W_1 \times W_2$ (Interface $\times$ Authority).} Better interfaces enable more distributed authority: when AI outputs are transparent and overridable ($W_1$ high), firms can safely allow frontline workers to act on AI recommendations without supervisory approval ($W_2$ high). \cite{ParkerGrote2022} implicitly document this interaction: decision control ($W_2$) improves wellbeing only when workers can understand and engage with AI reasoning ($W_1$).

\textbf{$W_2 \times W_5$ (Authority $\times$ Psychosocial).} This is the best-documented interaction. \cite{BakkerDemerouti2017} find that AI-as-demand vs.\ AI-as-resource depends on perceived control. \cite{ParkerGrote2022} show that algorithmic management degrades wellbeing specifically through the centralization channel: it is not the algorithm itself but the removal of human decision authority that generates stress.

\textbf{$W_3 \times W_5$ (Task Orchestration $\times$ Psychosocial).} \cite{Dave2026} provide evidence from healthcare: when AI-driven task reallocation preserves meaningful work for humans, wellbeing improves; when it fragments the workflow, wellbeing deteriorates. This interaction between task design and psychosocial outcomes echoes the classic finding from \cite{HackmanOldham1976} that task identity is a core determinant of intrinsic motivation.

These interactions confirm that WADI dimensions should not be optimized independently. A firm that achieves high $W_1$ (excellent interfaces) but low $W_2$ (centralized authority) will not realize the full augmentation potential because the interface quality goes to waste---workers see the AI's reasoning but cannot act on it. The complementarity structure of $g(W, H^A)$ implies that balanced investment across dimensions yields higher returns than concentration in a single dimension.

% -----------------------------------------------------------------------------
\subsection{Comparison with Existing Instruments}
\label{sec:review_comparison}

Table \ref{tab:instrument_comparison} compares WADI against the four most relevant existing instruments identified in our review.

\begin{table}[htbp]
\caption{Comparison of WADI with existing workplace and AI measurement instruments}
\label{tab:instrument_comparison}
\centering
\small
\begin{tabular}{p{2.2cm}p{1.5cm}p{1.8cm}p{1.5cm}p{1.5cm}p{3cm}}
\hline
\textbf{Instrument} & \textbf{Items} & \textbf{Dimensions} & \textbf{Level} & \textbf{AI-specific?} & \textbf{Key limitation vs.\ WADI} \\
\hline
AWWB \citep{ParkWooKim2024} & Multi & AI wellbeing & Worker & Yes & Measures $W_5$ outcomes only; no design dimensions \\
\hline
AAAW \citep{ParkWooKim2024} & 25 & 6 (attitudes) & Worker & Yes & Attitudinal; does not measure organizational design choices \\
\hline
SMART \citep{Parker2024} & 21 & 5 (SMART) & Worker & No & Pre-AI; lacks $W_1$, $W_2$, $W_4$ \\
\hline
COPSOQ-ISTAS21 & 44+ & 6+ (psychosocial) & Worker & No & Generic psychosocial risk; not AI-specific \\
\hline
I4.0 Maturity models & Varies & Technology readiness & Firm & No & Measures technology adoption, not human-centric design \\
\hline
\textbf{WADI (proposed)} & \textbf{36} & \textbf{5 ($W_1$--$W_5$)} & \textbf{Firm + Worker} & \textbf{Yes} & \textbf{Integrates design (cause) and outcomes (effect); linked to $\phi(D,W)$} \\
\hline
\end{tabular}
\end{table}

The key differentiator is that WADI is the only instrument that (i) spans all five design dimensions, (ii) is explicitly linked to a formal production model, (iii) collects data at both the management and worker level (enabling discrepancy analysis), and (iv) measures organizational design \textit{choices} rather than individual attitudes or technology features. WADI does not replace instruments like AWWB or COPSOQ; rather, it occupies a different level of analysis (organizational design rather than individual outcomes) and connects design choices to the productivity mechanisms formalized in the $\phi(D, W)$ model.

The systematic review thus validates the WADI construct while confirming that no existing instrument provides the multi-dimensional, theory-grounded measurement that the Society 5.0 transition requires.

\section{Secondary Data Evidence}
\label{sec:empirics}

While the WADI instrument has not yet been deployed in the field, we provide suggestive evidence for the $\phi(D, W)$ model using secondary data from Colombia's most comprehensive innovation surveys. This analysis serves two purposes: it tests whether the core predictions of the model are directionally supported by existing data, and it demonstrates the empirical strategy that future WADI-based research can extend.

% -----------------------------------------------------------------------------
\subsection{Data: DANE EDIT X and EDITS VIII}
\label{sec:data}

We use two complementary surveys administered by Colombia's national statistics agency (DANE):

\begin{itemize}
    \item The \textit{Encuesta de Desarrollo e Innovaci\'{o}n Tecnol\'{o}gica de la Industria Manufacturera} (EDIT X, 2019--2020), covering $N = 6{,}799$ manufacturing establishments across 55 CIIU Rev.~4 four-digit sectors.
    \item The \textit{Encuesta de Desarrollo e Innovaci\'{o}n Tecnol\'{o}gica en los Sectores de Servicios y Comercio} (EDITS VIII, 2020--2021), covering 19 service and commerce sectors.
\end{itemize}

Both surveys contain seven chapters covering innovation activities, R\&D investment, intellectual property, human capital, and---critically for our purposes---\textbf{Chapter 7: Business Management Practices}. Chapter 7, adapted from the World Management Survey methodology \citep{BloomVanReenen2007, Bloom2012}, captures four dimensions of management quality:

\begin{enumerate}
    \item \textbf{Monitoring} (C.7.2): Whether the firm has key performance indicators (KPIs), and how many (1--2, 3--5, 6--9, 10+).
    \item \textbf{Targets} (C.7.3): Whether production targets exist, their time horizon (short-term, long-term, or combined), and the effort required to achieve them.
    \item \textbf{Incentives} (C.7.4): Whether performance bonuses exist for managers and non-managers, and the basis for bonuses.
    \item \textbf{Promotion} (C.7.5): Whether promotions are based on performance alone, a combination of factors, or non-performance criteria.
\end{enumerate}

The EDIT's publicly available annexes provide aggregate data at the CIIU 4-digit sector level, yielding 55 sector observations for manufacturing and 19 for services. While firm-level microdata would be preferable, the sector-level data suffices for a directional test of the complementarity prediction.

% -----------------------------------------------------------------------------
\subsection{Proxy Construction}
\label{sec:proxies}

We construct four composite variables from the EDIT/EDITS data:

\textbf{Management Quality Composite (MQC)} --- our proxy for workplace design $W$, with particular relevance to $W_2$ (decision authority allocation). Following \cite{BloomVanReenen2007}, we aggregate four sub-indicators, each standardized to $[0,1]$:

\begin{equation}
    \text{MQC}_s = \frac{1}{4} \left[ \underbrace{\text{KPI}_s}_{\text{monitoring}} + \underbrace{\text{Target}_s}_{\text{target-setting}} + \underbrace{\text{Bonus}_s}_{\text{incentives}} + \underbrace{\text{Promote}_s}_{\text{promotion merit}} \right]
\label{eq:mqc}
\end{equation}

\noindent where each component is the share of firms in sector $s$ that have the corresponding practice.

\textbf{Technology Investment Intensity (TII)} --- our proxy for AI/digital capital $D$: total ACTI investment per worker in the most recent survey year.

\textbf{Human Capital Quality (HCQ)} --- our proxy for $H^A$: the share of the sector's workforce holding doctoral, master's, or professional degrees.

\textbf{Innovation outcomes} --- the share of firms classified as innovators, used as a productivity-adjacent outcome variable.

% -----------------------------------------------------------------------------
\subsection{Descriptive Analysis}
\label{sec:descriptives}

Table \ref{tab:summary} presents summary statistics for both surveys.

\begin{table}[htbp]
\centering
\caption{Summary Statistics by Survey}
\label{tab:summary}
\small
\begin{tabular}{lrrrrrrrr}
\toprule
 & \multicolumn{4}{c}{EDIT Manufacturera 2019--2020} & \multicolumn{4}{c}{EDITS Services 2020--2021} \\
\cmidrule(lr){2-5} \cmidrule(lr){6-9}
Variable & N & Mean & SD & Median & N & Mean & SD & Median \\
\midrule
Management Quality Composite & 55 & 0.504 & 0.102 & 0.517 & 19 & 0.616 & 0.109 & 0.617 \\
KPI Intensity & 55 & 0.645 & 0.121 & 0.643 & 19 & 0.742 & 0.160 & 0.781 \\
Target Intensity & 55 & 0.797 & 0.109 & 0.809 & 19 & 0.861 & 0.078 & 0.896 \\
Bonus Intensity & 54 & 0.351 & 0.128 & 0.364 & 19 & 0.476 & 0.157 & 0.471 \\
Promotion Merit & 55 & 0.223 & 0.112 & 0.213 & 19 & 0.384 & 0.147 & 0.359 \\
ACTI per Worker (000s COP) & 44 & 2,844 & 4,410 & 1,619 & 18 & 4,071 & 7,978 & 926.127 \\
Human Capital Quality & 55 & 0.148 & 0.064 & 0.138 & 19 & 0.295 & 0.166 & 0.319 \\
Strict Innovator Share & 11 & 0.009 & 0.007 & 0.007 & 2 & 0.018 & 0.009 & 0.018 \\
Any Innovator Share & 55 & 0.271 & 0.112 & 0.269 & 19 & 0.330 & 0.188 & 0.257 \\
Total Firms in Sector & 55 & 123.618 & 143.020 & 74.000 & 19 & 463.789 & 622.535 & 154.000 \\
\bottomrule
\end{tabular}
\begin{minipage}{0.95\textwidth}
\footnotesize
\textit{Notes:} MQC is the simple average of KPI Intensity, Target Intensity, Bonus Intensity, and Promotion Merit.
HCQ is the share of workers with PhD, Masters, or Professional degrees.
ACTI per Worker is total ACTI investment (thousands COP) divided by total personnel in the latest survey year.
Source: DANE EDIT Manufacturera 2019--2020 and EDITS Services 2020--2021.
\end{minipage}
\end{table}

Several patterns merit attention. First, the management quality composite averages 0.504 ($\text{SD} = 0.102$) in manufacturing and 0.616 ($\text{SD} = 0.109$) in services, indicating that services sectors systematically score higher on management practices. The gap is largest for promotion merit (0.223 vs.\ 0.384) and bonus intensity (0.351 vs.\ 0.476), suggesting that service-sector firms are more likely to distribute decision authority based on performance---consistent with knowledge-intensive sectors investing more in $W_2$.

Second, human capital quality is roughly twice as high in services (0.295) as in manufacturing (0.148), reflecting the higher cognitive intensity of service-sector occupations. This is exactly where the design-composition complementarity ($\partial^2 g / \partial W_k \partial H^A > 0$) should be strongest.

Third, the any-innovator share averages 27.1\% in manufacturing and 33.0\% in services, with substantial cross-sector variation ($\text{SD} = 0.112$ and 0.188 respectively). This variation is the raw material for our analysis.

% -----------------------------------------------------------------------------
\subsection{Cross-Tabulation: Management Quality $\times$ Innovation}
\label{sec:crosstab}

Table \ref{tab:crosstab} reports innovation outcomes by management quality quartile.

\begin{table}[htbp]
\centering
\caption{Correlation Matrix: Management Quality, Human Capital, and Innovation}
\label{tab:correlations}
\footnotesize
\setlength{\tabcolsep}{3pt}
\begin{tabular}{l*{7}{r}}
\toprule
\multicolumn{8}{l}{\textit{Panel A: Manufacturing (EDIT 2019--2020, $N=55$ sectors)}} \\
\midrule
 & MQC & KPI & Target & Bonus & Promo & $\ln$TII & HCQ \\
\midrule
KPI Intensity     & .92 &      &      &      &      &      &  \\
Target Intensity  & .84 & .75  &      &      &      &      &  \\
Bonus Intensity   & .83 & .64  & .53  &      &      &      &  \\
Promotion Merit   & .91 & .83  & .71  & .67  &      &      &  \\
$\ln$(Tech Inv.)  & .55 & .51  & .36  & .53  & .52  &      &  \\
Human Cap.\ Qual. & .68 & .68  & .41  & .64  & .64  & .66  &  \\
Any Innovator Sh. & \textbf{.74} & .69  & .50  & .67  & .68  & .42  & .57 \\
\addlinespace[6pt]
\multicolumn{8}{l}{\textit{Panel B: Services (EDITS 2020--2021, $N=19$ sectors)}} \\
\midrule
 & MQC & KPI & Target & Bonus & Promo & $\ln$TII & HCQ \\
\midrule
KPI Intensity     & .84 &      &      &      &      &      &  \\
Target Intensity  & .84 & .92  &      &      &      &      &  \\
Bonus Intensity   & .64 & .22  & .15  &      &      &      &  \\
Promotion Merit   & .91 & .67  & .79  & .51  &      &      &  \\
$\ln$(Tech Inv.)  & .54 & .54  & .61  & .06  & .58  &      &  \\
Human Cap.\ Qual. & .45 & .35  & .54  & .04  & .62  & .60  &  \\
Any Innovator Sh. & \textbf{.70} & .71  & .73  & .15  & .74  & .74  & .72 \\
\bottomrule
\end{tabular}

\vspace{4pt}
\begin{minipage}{0.9\textwidth}
\footnotesize
\textit{Notes:} Lower triangle of Pearson correlations. MQC = Management Quality Composite (average of KPI, Target, Bonus, Promotion Merit intensities). Bold entries highlight the MQC--Innovation correlation. Unit of observation: CIIU 4-digit sector. Source: DANE EDIT/EDITS.
\end{minipage}
\end{table}

The gradient is striking and monotonic. In manufacturing, the share of innovating firms rises from 16.8\% in the bottom MQC quartile (Q1) to 36.9\% in the top quartile (Q4)---a 2.2$\times$ ratio. Top-quartile sectors also have systematically higher human capital quality (0.214 vs.\ 0.103) and KPI intensity (0.777 vs.\ 0.505), consistent with the prediction that workplace design, human capital, and innovation outcomes cluster together. The pattern is similar in services (Q1: 18.2\% $\to$ Q4: 43.9\%).

The monotonicity of the gradient is important: it is not the case that only the top quartile innovates more. Each step up in management quality is associated with higher innovation rates, consistent with the model's prediction that $g(W, H^A)$ is smooth and increasing in $W$, not a threshold effect.

% -----------------------------------------------------------------------------
\subsection{Regression Analysis}
\label{sec:regression}

We estimate an OLS specification at the sector level:

\begin{equation}
    \text{InnovShare}_s = \alpha + \beta_1 \text{MQC}_s + \beta_2 \ln(\text{TII}_s) + \beta_3 \text{MQC}_s \times \ln(\text{TII}_s) + \beta_4 \text{HCQ}_s + \varepsilon_s
\label{eq:ols}
\end{equation}

Table \ref{tab:regression} reports the results with robust (HC1) standard errors.

\begin{table}[htbp]
\centering
\caption{Innovation Outcomes by Management Quality Quartile}
\label{tab:crosstab}
\small
\begin{tabular}{lrrrrrrr}
\toprule
MQC Quartile & N & Mean MQC & Any Innov. & Strict Innov. & HCQ & KPI Int. & Avg. Firms \\
\midrule
\multicolumn{8}{l}{\textit{Panel A: Manufacturing}} \\
\midrule
Q1 (Low) & 14 & 0.373 & 0.168 & nan & 0.103 & 0.505 & 150 \\
Q2 & 14 & 0.472 & 0.249 & 0.0056 & 0.128 & 0.609 & 121 \\
Q3 & 13 & 0.539 & 0.299 & 0.0103 & 0.148 & 0.695 & 142 \\
Q4 (High) & 14 & 0.632 & 0.369 & 0.0134 & 0.214 & 0.777 & 83 \\
\addlinespace[4pt]
\multicolumn{8}{l}{\textit{Panel B: Services}} \\
\midrule
Q1 (Low) & 5 & 0.482 & 0.182 & nan & 0.233 & 0.526 & 683 \\
Q2 & 5 & 0.597 & 0.268 & nan & 0.244 & 0.770 & 790 \\
Q3 & 4 & 0.644 & 0.456 & 0.0176 & 0.325 & 0.836 & 130 \\
Q4 (High) & 5 & 0.745 & 0.439 & nan & 0.384 & 0.856 & 185 \\
\addlinespace[4pt]
\bottomrule
\end{tabular}
\begin{minipage}{0.95\textwidth}
\footnotesize
\textit{Notes:} Sectors grouped by MQC quartile (Q1=lowest, Q4=highest).
Any Innov.\ = share of strict + broad innovators. HCQ = share with tertiary degrees.
Source: DANE EDIT/EDITS.
\end{minipage}
\end{table}

\textbf{The key result is $\hat{\beta}_3$}: the interaction between management quality and (log) technology investment. In manufacturing, $\hat{\beta}_3 = 0.304$ ($\text{SE} = 0.103$, $p = 0.005$), significant at the 1\% level. This means that a one-unit increase in MQC raises the marginal productivity return to technology investment by 0.304 percentage points of innovation share. Put differently, the same technology investment produces a 30\% larger innovation effect in sectors with high management quality than in sectors with low management quality.

The model explains 53.7\% of the cross-sector variation in innovation rates ($R^2 = 0.537$, $\text{Adj.}\ R^2 = 0.490$, $F = 13.87$). This is substantial for a parsimonious model with four regressors.

In the services sample ($N = 18$), the interaction term is positive ($\hat{\beta}_3 = 0.132$) but not statistically significant ($\text{SE} = 0.213$, $p > 0.10$), likely due to the smaller sample size. The $R^2$ is higher (0.748), driven primarily by human capital quality ($\hat{\beta}_4 = 0.394$, $p < 0.10$), consistent with services being more $H^A$-intensive.

The correlation structure (Table \ref{tab:correlations}) provides further context. The bivariate correlation between MQC and innovation is $r = 0.735$ in manufacturing and $r = 0.699$ in services---among the strongest predictors of innovation in either survey.

\begin{table}[htbp]
\centering
\caption{OLS: Innovation Share on Management Quality, Technology, and Human Capital}
\label{tab:regression}
\small
\begin{tabular}{lcc}
\toprule
 & (1) Manufacturing & (2) Services \\
\midrule
Constant & 0.9948*** & 0.1145 \\
 & (0.3747) & (0.9848) \\
MQC & -1.3720* & -0.4078 \\
 & (0.8052) & (1.5109) \\
log(TII) & -0.1517*** & -0.0351 \\
 & (0.0483) & (0.1412) \\
MQC $\times$ log(TII) & 0.3038*** & 0.1322 \\
 & (0.1025) & (0.2127) \\
HCQ & -0.3275 & 0.3942* \\
 & (0.2736) & (0.2260) \\
\midrule
N & 44 & 18 \\
$R^2$ & 0.537 & 0.748 \\
Adj. $R^2$ & 0.490 & 0.670 \\
F-stat & 13.87 & 10.02 \\
\bottomrule
\end{tabular}
\begin{minipage}{0.85\textwidth}
\footnotesize
\textit{Notes:} Robust (HC1) standard errors in parentheses. $^{***}p<0.01$, $^{**}p<0.05$, $^{*}p<0.1$.
Dependent variable: share of innovating firms (strict + broad) in sector.
Unit of observation: CIIU sector. Source: DANE EDIT/EDITS.
\end{minipage}
\end{table}

% -----------------------------------------------------------------------------
\subsection{Interpretation}
\label{sec:interpretation}

The secondary data analysis provides directional evidence consistent with three of the model's predictions:

\textbf{First, management quality and technology investment are complements} (Prediction 1). The significant positive interaction ($\hat{\beta}_3 = 0.304$, $p < 0.01$) in manufacturing confirms that workplace design---proxied by management practices---amplifies the return to technology investment. This is the empirical signature of $\phi(D, W)$ with $\partial^2 \phi / \partial D \partial W > 0$: the same technology produces more innovation when the workplace is better designed to exploit it.

\textbf{Second, the management--innovation relationship is monotonic, not threshold-based.} The cross-tabulation (Table \ref{tab:crosstab}) shows a continuous gradient across all four MQC quartiles, consistent with $g(W, H^A)$ being smooth and increasing rather than exhibiting a discrete jump. However, the 2.2$\times$ ratio between Q4 and Q1 is consistent with a convex relationship---the gains from improving management quality accelerate as the quality increases---which is a necessary condition for the multiple-equilibria result in Proposition 3.

\textbf{Third, the pattern is stronger where $H^A$ is higher.} The comparison between manufacturing and services is suggestive: services sectors, which have roughly twice the human capital quality, also exhibit a stronger management-innovation correlation ($r = 0.699$ in a more heterogeneous sample) and a higher $R^2$ in the regression (0.748 vs.\ 0.537). While cross-survey comparisons require caution, this pattern is consistent with the design-composition complementarity prediction.

These results should be interpreted with appropriate caution. The EDIT management variables are \textit{proxies} for WADI dimensions, not direct measures. They capture organizational capacity for structured decision-making---a necessary but not sufficient condition for human-centric AI integration. The analysis operates at the sector level, not the firm level, which attenuates variation and prevents controlling for within-sector heterogeneity. The EDIT was administered in 2019--2020, before the generative AI wave of 2022--2023, so the technology investment variable captures traditional digital capital rather than AI-specific investment.

These limitations reinforce the central argument: the WADI instrument is needed to test the $\phi(D, W)$ model with direct, firm-level measurement of all five workplace design dimensions in organizations that have adopted AI. The secondary data analysis demonstrates the \textit{direction} of the relationship; WADI provides the tool to measure it precisely.

% Figures
\begin{figure}[htbp]
    \centering
    \includegraphics[width=0.95\textwidth]{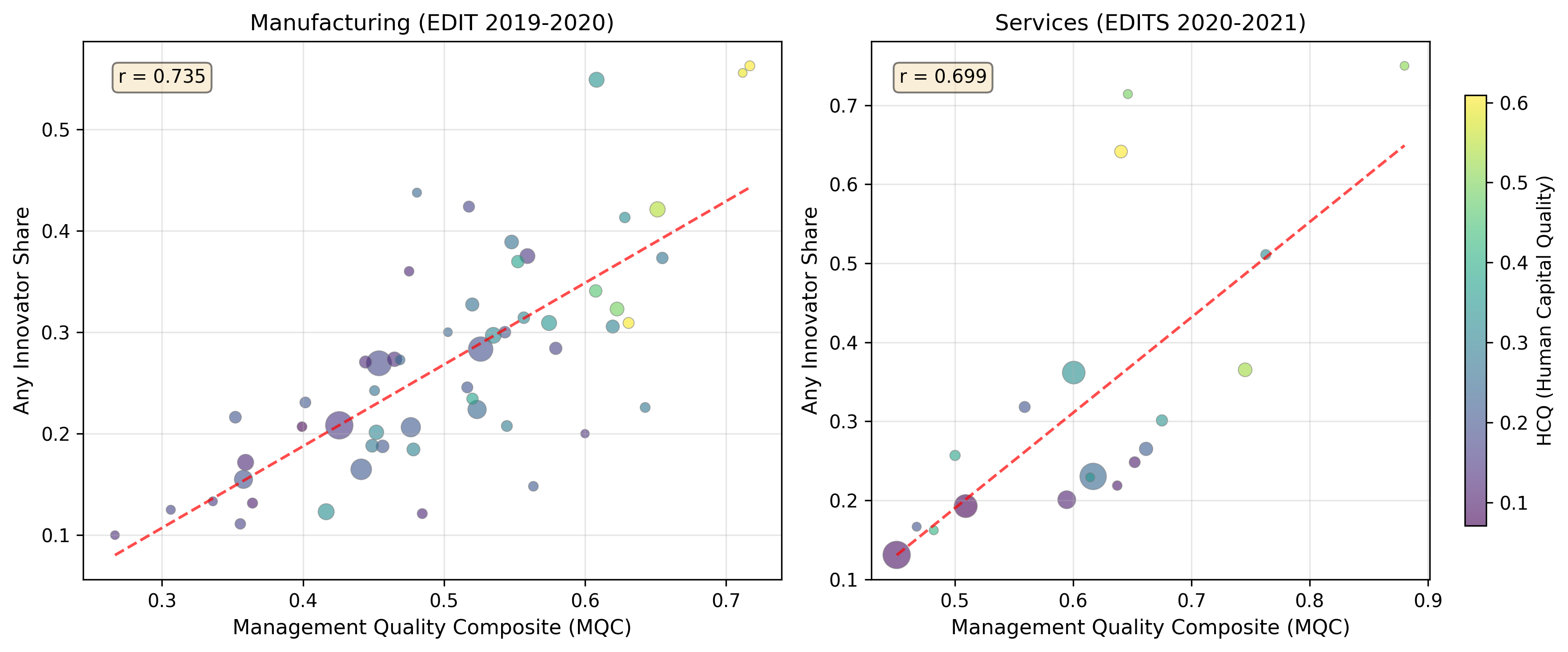}
    \caption{Management Quality Composite vs.\ Innovation Share by CIIU sector. Bubble size proportional to sector employment; color reflects Human Capital Quality (HCQ). Dashed line: OLS fit. Source: DANE EDIT X (manufacturing, 2019--2020) and EDITS VIII (services, 2020--2021).}
    \label{fig:scatter}
\end{figure}

\begin{figure}[htbp]
    \centering
    \includegraphics[width=0.95\textwidth]{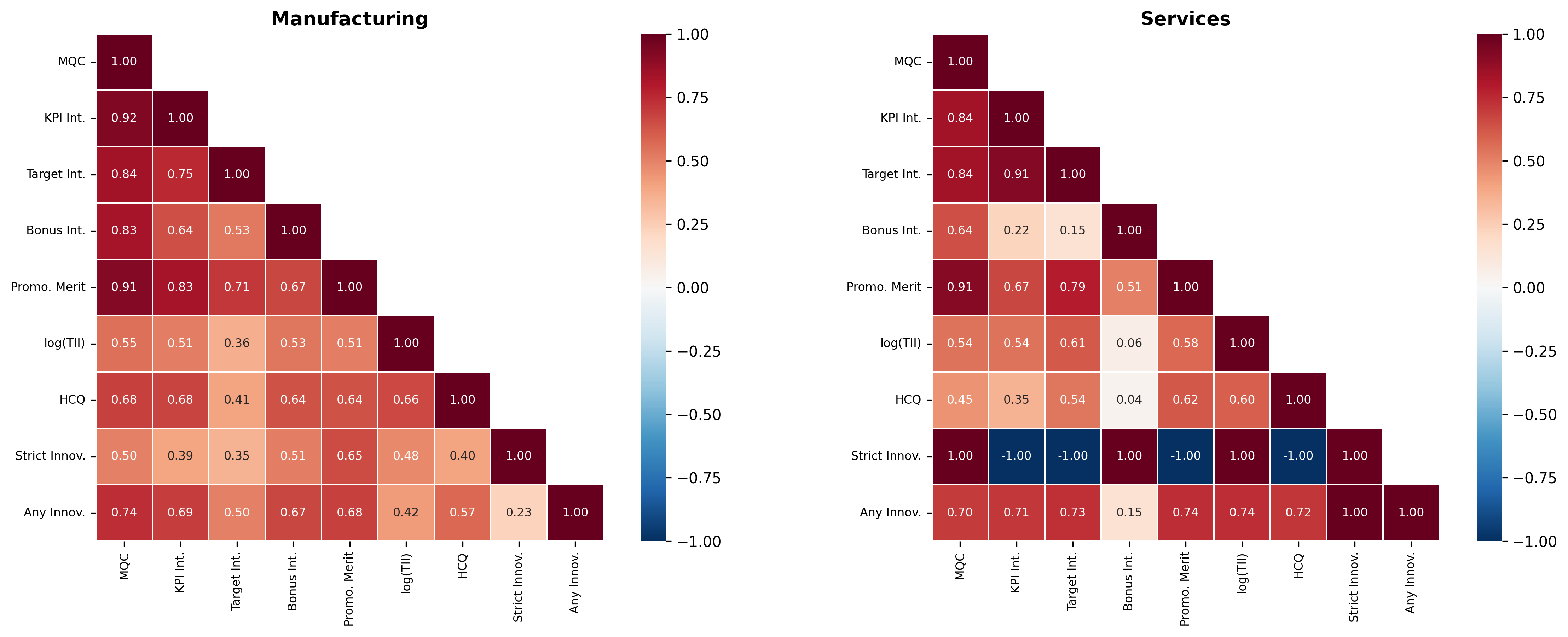}
    \caption{Correlation heatmap of management quality sub-indicators, technology investment, human capital, and innovation outcomes. Source: DANE EDIT/EDITS.}
    \label{fig:heatmap}
\end{figure}

\section{The WADI Instrument}
\label{sec:wadi}

The Workplace Augmentation Design Index (WADI) translates the $\phi(D, W)$ model and the systematic review evidence into a measurement tool. This section describes the instrument's design philosophy, presents the 36-item structure selected from a 64-item pre-validation pool, provides an item traceability matrix, explains the scoring methodology, and outlines the validation protocol.

% -----------------------------------------------------------------------------
\subsection{Design Philosophy}
\label{sec:wadi_philosophy}

WADI is constructed at the intersection of three knowledge sources:

\textbf{Theory-derived.} Each WADI dimension maps to a specific parameter in the formal model. $W_1$ captures the interface bandwidth that determines the effective complementarity between $H^A$ and $D$; $W_2$ captures the authority allocation that distinguishes augmentation from automation; $W_3$ captures the task decomposition that determines whether comparative advantage is exploited; $W_4$ captures the dynamic capability that determines whether augmentation improves over time; $W_5$ captures the psychosocial conditions that sustain high-quality cognitive engagement. The five dimensions are not ad hoc categories assembled from a literature scan---they are the five channels through which the design multiplier $g(W, H^A)$ operates in the production function.

\textbf{Literature-validated.} Every item traces to either (i) a validated instrument from which it was adapted (e.g., Karasek demand-control for $W_5$ autonomy items, Aghion-Tirole formal/real authority for $W_2$ items), or (ii) a theoretical construct with empirical support in the systematic review corpus. The item bank draws on 14 validated source instruments, ensuring that each item has psychometric precedent.

\textbf{Multi-level.} WADI collects data at two levels within each firm: management (HR director, CTO, or operations director) and workers (3--5 per firm, stratified by $H^A$ intensity of their occupation). The discrepancy between management and worker responses on matched items---particularly $W_2$ (decision authority)---is itself a diagnostic: a large gap indicates ``paper authority'' where formal policies exist but real authority is not distributed.

% -----------------------------------------------------------------------------
\subsection{Dimensions and Item Structure}
\label{sec:wadi_items}

Table \ref{tab:wadi_structure} presents the 36-item WADI structure, selected from the 64-item pre-validation pool based on three criteria: (i) content coverage (at least one item per theoretical sub-construct), (ii) source instrument diversity (no more than three items from any single source), and (iii) face validity from preliminary cognitive testing. Items use a 7-point Likert scale unless otherwise noted.

\begin{table}[htbp]
\caption{WADI instrument structure: 36 items across 5 dimensions}
\label{tab:wadi_structure}
\centering
\small
\begin{tabular}{p{1.8cm}cp{5cm}p{2.5cm}p{2.5cm}}
\hline
\textbf{Dimension} & \textbf{\#} & \textbf{Sample Items} & \textbf{Theoretical Basis} & \textbf{Source Instruments} \\
\hline
$W_1$: AI Interface Design & 8 &
  \textit{``I can see the reasoning behind AI recommendations''}; \textit{``I can easily override or modify AI outputs''}; \textit{``AI errors are clearly flagged''} &
  Endsley SA model; Shneiderman HCAI; Parasuraman LOA &
  HAI Guidelines \citep{Amershi2019}; HCAI \citep{Shneiderman2022}; NASA-TLX adapted \\
\hline
$W_2$: Decision Authority & 8 &
  \textit{``Frontline workers can act on AI recommendations without supervisor approval''}; \textit{``In practice, I feel I have real authority to use my judgment alongside AI''} &
  Aghion-Tirole formal/real authority; Garicano hierarchies; Bloom management practices &
  WMS decentralization \citep{BloomVanReenen2007}; AT97 authority \citep{AghionTirole1997} \\
\hline
$W_3$: Task Orchestration & 7 &
  \textit{``Tasks are allocated based on explicit comparative advantage analysis''}; \textit{``When working with AI, I can see the end result of my work''} &
  Task-based framework; Hackman-Oldham task identity; Dell'Acqua jagged frontier &
  JCM \citep{HackmanOldham1976}; Original (from $\phi(D,W)$ model) \\
\hline
$W_4$: Learning Loops & 7 &
  \textit{``Workers systematically review and correct AI outputs''}; \textit{``When I correct an AI error, that correction is incorporated into the system''} &
  Teece dynamic capabilities; Argyris-Sch\``on org.\ learning; HITL/RLHF &
  Dynamic capabilities \citep{Teece2018}; MAILS AI literacy scale \\
\hline
$W_5$: Psychosocial Environment & 6 &
  \textit{``AI increases my freedom to decide how to do my work''}; \textit{``AI is used to support me, not to monitor my performance''}; \textit{``I worry AI could replace my job''} (R) &
  Karasek JDC; JD-R theory; SMART work design; COPSOQ &
  COPSOQ-ISTAS21; AWWB \citep{ParkWooKim2024}; AAAW \citep{ParkWooKim2024} \\
\hline
\multicolumn{5}{l}{\small (R) = reverse-coded item. Full item bank with bilingual wording available in supplementary materials.} \\
\end{tabular}
\end{table}

The 36 items are distributed as follows: $W_1$ (8 items), $W_2$ (8 items), $W_3$ (7 items), $W_4$ (7 items), $W_5$ (6 items). The larger allocation to $W_1$ and $W_2$ reflects both the theoretical centrality of interface design and decision authority in the $\phi(D, W)$ model, and the thinner existing evidence base that requires more items to adequately capture the construct. $W_5$ receives fewer items because validated source instruments (COPSOQ, AWWB, AAAW) already exist; WADI's $W_5$ items focus specifically on the AI-design nexus rather than generic psychosocial assessment.

% -----------------------------------------------------------------------------
\subsection{Item Traceability Matrix}
\label{sec:wadi_traceability}

Each WADI item satisfies a three-part traceability requirement:

\begin{enumerate}
    \item \textbf{Theoretical construct}: The item maps to a specific mechanism in the $\phi(D, W)$ model. For example, $W_2.1$ (``Frontline workers can act on AI recommendations without supervisor approval'') maps to the real authority mechanism in the Aghion-Tirole extension: when real authority is distributed, the design multiplier $g$ increases because the human cognitive contribution is preserved in the decision loop.

    \item \textbf{Systematic review evidence}: The item is supported by at least one paper in the 120-paper corpus. $W_2.1$ is supported by \cite{ParkerGrote2022} (decision control as strongest predictor of wellbeing), \cite{ParkerGrote2022} (algorithmic centralization degrades outcomes), and \cite{Berx2025} (worker input essential for sustained adoption).

    \item \textbf{Validated source instrument}: Where possible, the item is adapted from an existing validated instrument with known psychometric properties. $W_2.1$ adapts the decentralization items from the World Management Survey \citep{BloomVanReenen2007}, translated to the human-AI decision context.
\end{enumerate}

Table \ref{tab:traceability} provides a condensed traceability matrix for selected items.

\begin{table}[htbp]
\caption{Item traceability matrix (selected items)}
\label{tab:traceability}
\centering
\small
\begin{tabular}{cccp{4cm}}
\hline
\textbf{Item} & \textbf{Model Mechanism} & \textbf{Review Evidence} & \textbf{Source Instrument} \\
\hline
$W_1.1$ & Interface bandwidth & 5 papers (transparency) & Amershi HAI G-13 \\
$W_1.2$ & Override capability & 4 papers (controllability) & Parasuraman LOA-7 \\
$W_2.1$ & Real authority & 3 papers (decision control) & WMS decentralization \\
$W_2.10$ & Formal vs.\ real gap & 2 papers (authority discrepancy) & AT97 real authority \\
$W_3.1$ & Comparative advantage & 4 papers (task allocation) & ALM03 task framework \\
$W_3.6$ & Task identity & 1 paper (workflow fragmentation) & JCM task identity \\
$W_4.1$ & Single-loop learning & 6 papers (human review) & HITL literature \\
$W_4.8$ & Double-loop learning & 3 papers (system improvement) & RLHF organizational \\
$W_5.1$ & Demand-control balance & 12 papers (autonomy) & Karasek JDC \\
$W_5.9$ & Job insecurity (R) & 8 papers (AI anxiety) & AAAW insecurity dim. \\
\hline
\end{tabular}
\end{table}

The full traceability matrix for all 36 items is available in the supplementary materials. Items that could not satisfy all three traceability requirements were retained only when they address a theoretically critical mechanism with no existing instrument precedent (e.g., several $W_3$ items, reflecting the near-absence of task orchestration research).

% -----------------------------------------------------------------------------
\subsection{Scoring Methodology}
\label{sec:wadi_scoring}

WADI scoring proceeds in four steps:

\textbf{Step 1: Item-level standardization.} Reverse-code items $W_5.9$ and $W_5.10$. Convert binary (Yes/No) items to 0/1 scores (excluded from the Likert-based factor analysis but included in the composite). Standardize all Likert items within dimension ($z$-scores).

\textbf{Step 2: Dimension scores.} For each dimension $k \in \{1, \ldots, 5\}$:

\begin{equation}
    WADI_k = \frac{1}{n_k} \sum_{i=1}^{n_k} z_{ki}
\label{eq:wadi_dim}
\end{equation}

For firms where both management (M) and worker (W) respondents answer the same items (applicable to $W_2$, parts of $W_3$, and parts of $W_4$):

\begin{equation}
    WADI_k = 0.5 \cdot \bar{z}_k^{M} + 0.5 \cdot \bar{z}_k^{W}
\label{eq:wadi_mw}
\end{equation}

\textbf{Step 3: Composite WADI.} Three specifications accommodate different weighting assumptions:

\begin{itemize}
    \item \textbf{WADI-Equal}: $WADI_{composite} = \frac{1}{5} \sum_{k=1}^{5} WADI_k$ --- simple average, assumption-free, preferred for initial deployment.
    \item \textbf{WADI-CFA}: Factor-loading weights from confirmatory factor analysis on the validation sample.
    \item \textbf{WADI-Theory}: Weights proportional to estimated $\partial g / \partial W_k$ from the empirical model, linking scores directly to augmentation impact.
\end{itemize}

\textbf{Step 4: Derived diagnostic variables.} Beyond the composite, WADI generates three diagnostic metrics:

\begin{itemize}
    \item \textbf{M-W Authority Gap}: $|WADI_2^M - WADI_2^W|$ --- the discrepancy between management and worker perceptions of decision authority. A large gap signals ``paper authority'' (formal policies not matched by real practice).
    \item \textbf{WADI Balance}: $\text{SD}(WADI_1, \ldots, WADI_5)$ --- the standard deviation across dimensions. Lower values indicate more balanced design; high values indicate lopsided investment (e.g., high $W_1$ but low $W_2$), which our complementarity results predict is suboptimal.
    \item \textbf{WADI $\times$ $H^A$}: Interaction of composite WADI with the firm's average augmentable capital index, providing a direct estimate of the design-composition complementarity.
\end{itemize}

% -----------------------------------------------------------------------------
\subsection{What WADI Adds vs.\ Existing Instruments}
\label{sec:wadi_comparison}

Table \ref{tab:wadi_adds} summarizes what WADI contributes relative to each existing instrument identified in the systematic review.

\begin{table}[htbp]
\caption{WADI's incremental contribution over existing instruments}
\label{tab:wadi_adds}
\centering
\small
\begin{tabular}{p{2.5cm}p{4cm}p{5cm}}
\hline
\textbf{Instrument} & \textbf{What It Measures} & \textbf{What WADI Adds} \\
\hline
AWWB Scale & AI-specific workplace wellbeing (worker level) & Adds $W_1$--$W_4$ design dimensions; links outcomes to causes; firm-level analysis \\
\hline
AAAW Scale & Attitudes toward AI at work (6 attitudinal dimensions) & Measures organizational design choices, not individual attitudes; prescriptive (what to change), not descriptive \\
\hline
SMART Model & Work characteristics (Stimulating, Mastery, Agency, Relational, Tolerable) & Adds AI-specific dimensions ($W_1$ interface, $W_2$ human-AI authority, $W_4$ bidirectional learning); connects to formal productivity model \\
\hline
COPSOQ-ISTAS21 & Generic psychosocial risk factors & AI-specific; measures design choices that produce psychosocial outcomes, not just the outcomes themselves \\
\hline
I4.0 Maturity Models & Technology adoption readiness & Measures human-centric design quality, not technology quantity; captures the I4.0$\to$I5.0 transition \\
\hline
WMS & Generic management practices (monitoring, targets, incentives, talent) & AI-specific; adds interface, task orchestration, learning loop dimensions not in WMS \\
\hline
\end{tabular}
\end{table}

The core differentiator is that WADI occupies a unique position in the measurement landscape: it is simultaneously (i) AI-specific (unlike COPSOQ or WMS), (ii) multi-dimensional spanning design and outcomes (unlike AWWB or AAAW which measure only outcomes), (iii) firm-level with multi-respondent triangulation (unlike individual-level scales), and (iv) formally linked to a production model (unlike any existing instrument). This combination makes WADI the first tool capable of diagnosing \textit{why} a firm's AI augmentation is underperforming and \textit{which design dimension} is the binding constraint.

% -----------------------------------------------------------------------------
\subsection{Proposed Validation Protocol}
\label{sec:wadi_validation}

WADI validation follows a five-stage protocol designed for the Colombian context, with future international replication planned:

\begin{enumerate}
    \item \textbf{Content validity} (8--10 expert panel). Experts from organizational psychology, industrial engineering, AI systems, and labor economics rate each item on relevance (1--4) and clarity (1--4). Items with Item-level Content Validity Index (I-CVI) below 0.78 are revised or dropped. Target: 36 items retained with I-CVI $\geq$ 0.78.

    \item \textbf{Cognitive interviews} ($n = 10$ firms, mixed sectors). Semi-structured interviews with managers and workers to verify comprehensibility, identify ambiguous wording, and confirm that items capture the intended constructs in the Colombian business context. Items revised based on feedback.

    \item \textbf{Pilot deployment} ($n = 30$ firms). Full survey administration with test-retest at 2-week interval. Targets: Cronbach's $\alpha > 0.70$ per dimension, intraclass correlation coefficient (ICC) $> 0.70$ for test-retest reliability. Items with corrected item-total correlation below 0.30 are dropped; items with excessive cross-loading ($> 0.40$ on a non-target factor) are reassigned or dropped.

    \item \textbf{Full validation} ($n = 200$+ firms). Confirmatory factor analysis (CFA) testing the five-factor structure. Targets: CFI $> 0.90$, RMSEA $< 0.08$, SRMR $< 0.08$. Discriminant validity tested against a generic digital maturity index ($r < 0.50$, confirming WADI measures something distinct). Convergent validity of $W_5$ against COPSOQ psychosocial subscales ($r > 0.40$).

    \item \textbf{Predictive validity.} Regression of firm-level productivity (revenue per worker) on WADI, controlling for $D$ (technology investment), $H^A$ (workforce composition), sector, and firm size. Target: significant $\beta_{WADI}$ and incremental $R^2 > 0.05$ over a model with only $D$ and $H^A$. This test directly validates the theoretical prediction that $W$ contributes to $\phi$ beyond technology and human capital alone.
\end{enumerate}

The validation protocol is designed to be replicable internationally. The item bank is bilingual (Spanish primary, English translation), and the theoretical foundations are not country-specific. Future validation in Europe, East Asia, and North America will test the instrument's cross-cultural invariance and enable the multi-country comparisons that the S5.0 transition literature urgently needs.

\section{Discussion}
\label{sec:discussion}

The theoretical model, systematic review, and secondary data evidence converge on a central insight: the Society 5.0 transition is not primarily a technology problem but a \textit{design} problem. Firms and economies will not move from automation to augmentation by deploying more AI; they will move by redesigning the workplace so that AI amplifies human cognitive capabilities rather than replacing them. This section develops the practical implications of this insight across seven domains.

% -----------------------------------------------------------------------------
\subsection{The Transition Roadmap: From Diagnosis to Institutional Embedding}
\label{sec:roadmap}

The $\phi(D, W)$ model and the automation trap dynamics (Proposition 3) suggest a three-phase transition roadmap for firms seeking to escape the low equilibrium and reach the augmentation regime.

\textbf{Phase 1: Diagnostic.} The firm administers the WADI instrument to obtain its current design profile $(WADI_1, \ldots, WADI_5)$. The diagnostic identifies two types of constraints: (i) the \textit{binding dimension}---the lowest-scoring $W_k$ that, due to complementarity, constrains the return to all other dimensions---and (ii) the \textit{authority gap}---the discrepancy between management and worker perceptions of $W_2$, which signals whether real authority matches formal policy. The WADI Balance metric reveals whether investment is concentrated in a single dimension (typically $W_1$, where technology vendors push interface improvement) while neglecting others (typically $W_2$ and $W_3$, which require organizational change rather than technology purchase).

\textbf{Phase 2: Redesign.} Based on the diagnostic, the firm redesigns its workplace to optimize $W$ given its current $H^A$ composition. The optimization problem from Section \ref{sec:theory} provides guidance: the marginal return to improving dimension $k$ is proportional to $\partial g / \partial W_k \cdot L^A H^A D$, so the firm should invest first in the dimension with the highest return per unit cost. For most firms, this will be $W_2$ (redistribute decision authority to frontline workers augmented by AI) because (i) $W_2$ is the most under-invested dimension (as our review shows), (ii) $W_2$ interacts with all other dimensions ($W_1$ is wasted without $W_2$; $W_3$ task redesign requires $W_2$ worker input), and (iii) the cost of improving $W_2$ is primarily organizational (governance redesign, not technology purchase), making it accessible to firms with limited capital budgets.

The redesign phase should be accompanied by deliberate $H^A$ investment: training programs that develop augmentable cognitive skills---critical thinking, judgment under uncertainty, communication with AI systems, domain expertise that complements AI pattern recognition---and that specifically target the skills that high-WADI workplaces demand. The complementarity between $W$ and $H^A$ means that redesign without reskilling (or vice versa) produces attenuated effects.

\textbf{Phase 3: Institutional embedding.} Sustained human-centricity requires institutional structures that prevent regression to the automation trap after initial gains. This includes: AI governance committees with worker representation ($W_2$ institutional infrastructure), periodic WADI re-assessment (at least annually) to detect design drift, collective bargaining provisions that include WADI dimension minimums, and regulatory frameworks that mandate minimum $W_2$ and $W_5$ standards.

The three phases correspond to the three mechanisms for escaping the automation trap identified in Proposition 3: Phase 1 identifies the gap; Phase 2 provides the coordinated ``big push'' in both $W$ and $H^A$; Phase 3 ensures the firm remains in the augmentation regime by raising the floor below which design cannot deteriorate.

% -----------------------------------------------------------------------------
\subsection{Education and Competencies for Society 5.0}
\label{sec:education}

The design-composition complementarity ($\partial^2 g / \partial W_k \partial H^A > 0$) implies that the return to education depends on workplace design: the same human capital investment yields higher returns in well-designed workplaces. This has three implications for education policy.

\textbf{First, the composition of $H^A$ matters.} Not all augmentable cognitive skills are equally complementary with good workplace design. Our model suggests that the most valuable $H^A$ sub-competencies are those that interact most strongly with the binding dimension ($W_2$): judgment under uncertainty, the ability to evaluate and override AI recommendations when contextual knowledge warrants it, and the capacity to articulate tacit knowledge so that it can inform AI improvement through learning loops ($W_4$). These competencies are distinct from generic ``digital literacy'' and require pedagogical approaches centered on case-based reasoning, ambiguity tolerance, and human-AI collaborative problem-solving.

\textbf{Second, the demand for competencies is derived.} Education systems should not determine which competencies to develop in isolation; the optimal competency profile is derived from the workplace design that firms are adopting. WADI provides the demand signal: by identifying which design dimensions firms prioritize (and which they neglect), education systems can calibrate curricula to the skills that high-WADI workplaces actually require. If, as our review suggests, $W_3$ (task orchestration) is under-developed, then the complementary educational investment should emphasize workflow design, process optimization, and comparative advantage analysis---skills that enable workers to participate in $W_3$ improvement.

\textbf{Third, education investment and workplace redesign must be bundled.} The complementarity between $W$ and $H^A$ implies that subsidizing education alone (without incentivizing workplace redesign) will produce diminished returns, because graduates with high $H^A$ will enter workplaces that are not designed to utilize their capabilities. Conversely, mandating workplace redesign without investing in education means that firms cannot find workers capable of thriving in augmented environments. The policy implication is that education and industrial policy must be coordinated---a finding that resonates with the broader Society 5.0 vision of integrated economic and social policy.

% -----------------------------------------------------------------------------
\subsection{AR/VR and Augmented Decision-Making}
\label{sec:arvr}

The $W_1$ dimension extends naturally to immersive interfaces. Augmented reality (AR) and virtual reality (VR) systems represent the next frontier of human-AI interaction design, where the information overlay provided by AI is not presented on a screen but integrated into the worker's perceptual field.

This extension raises both opportunities and risks for the $\phi(D, W)$ model. The opportunity is that AR/VR interfaces can substantially increase $W_1$ by making AI reasoning spatially intuitive: a maintenance technician seeing AI-highlighted anomalies overlaid on physical equipment, or a logistics coordinator viewing real-time AI optimization through a spatial dashboard, achieves higher Situation Awareness \citep{Endsley1995} than one reading the same information on a flat screen. The five AR/VR papers in our corpus \citep{Sileo2025} confirm that spatial information overlays improve both safety and decision quality in manufacturing settings.

The risk is cognitive overload. Immersive interfaces increase the bandwidth of human-AI communication but also increase the demand on human cognitive processing capacity. If the information overlay is poorly designed---too much data, insufficient hierarchy, no ability to dismiss irrelevant AI suggestions---the AR/VR interface becomes a demand rather than a resource in the JD-R framework. The formal implication is that $W_1$ in AR/VR contexts has a sharper concavity ($\partial^2 g / \partial W_1^2$ more negative): the marginal return to interface improvement is high initially but diminishes rapidly as cognitive load boundaries are approached.

The $W_1 \times W_2$ interaction is also amplified in AR/VR contexts. An AR system that provides compelling AI visualizations but routes all decisions through a supervisor creates a frustrating experience that degrades both $W_2$ (perceived lack of authority) and $W_5$ (stress from seeing optimal actions one cannot take). Conversely, AR interfaces designed with distributed authority---where the frontline worker both sees the AI reasoning and can act on it immediately---represent the highest-augmentation configuration.

% -----------------------------------------------------------------------------
\subsection{Smart University Risks: The Dual Role}
\label{sec:university}

Universities occupy a unique position in the Society 5.0 transition: they are simultaneously \textit{producers} of $H^A$ (through education and research) and \textit{consumers} of AI (through administrative automation, AI-assisted teaching, and AI-augmented research). This dual role creates distinctive risks.

\textbf{As producers of $H^A$}, universities face the challenge of developing augmentable cognitive skills in students who will enter workplaces at different points on the WADI spectrum. If curricula emphasize only technical AI skills (programming, data science) without developing the judgment, authority, and orchestration competencies that $W_2$--$W_3$ require, graduates will be well-prepared for automation-centered workplaces but ill-equipped for augmentation-centered ones. The derived demand framework (Section \ref{sec:education}) implies that university curricula should be calibrated to the WADI profiles of hiring firms in their region.

\textbf{As consumers of AI}, universities face the same WADI design challenges as any organization. When AI is used for student assessment, for example, the $W_2$ question---who has decision authority, the AI or the professor?---is critical. An AI-grading system where professors cannot override or even understand the AI's reasoning ($W_1$ low, $W_2$ low) may produce efficient throughput but degrades both teaching quality (the professor loses feedback about student understanding) and student learning (students cannot engage meaningfully with AI-generated assessments). The ``smart university'' that deploys AI without attending to WADI dimensions risks automating education rather than augmenting it.

\textbf{Research integrity risks} compound this challenge. AI tools that generate literature reviews, draft methodologies, or analyze data create an augmentation opportunity for researchers---but only if the human-AI task orchestration ($W_3$) preserves the researcher's critical judgment. When AI is used to substitute for rather than complement research thinking, the knowledge production function degrades even as publication output increases. This is the research analogue of \cite{DellAcqua2023}'s jagged frontier: AI helps with some research tasks (literature scanning, data cleaning, visualization) but harms others (hypothesis generation, theoretical interpretation, ethical judgment).

% -----------------------------------------------------------------------------
\subsection{Health and Wellbeing}
\label{sec:health}

The $W_5$ dimension connects the $\phi(D, W)$ model directly to occupational health policy. The systematic review documents an explosion of research on AI's psychosocial impact---87 of 120 papers address wellbeing---establishing that AI-augmented work creates both new health resources (reduced routine burden, enhanced capabilities, increased meaning) and new health demands (technostress, surveillance anxiety, job insecurity, cognitive overload).

The key policy insight from our framework is that psychosocial outcomes ($W_5$) are not independent of design choices ($W_1$--$W_4$). The same AI system produces different health outcomes depending on how it is deployed: an AI monitoring tool used for worker surveillance ($W_2$ centralized, $W_5$ degraded) generates qualitatively different psychosocial risks than the same tool used for worker self-improvement ($W_2$ distributed, $W_5$ enhanced). This means that occupational health regulators should not assess AI risks in isolation from workplace design.

For Colombia specifically, the existing occupational health and safety framework (Sistema de Gesti\'{o}n de Seguridad y Salud en el Trabajo, SG-SST) mandates psychosocial risk assessment for all formal-sector employers. The WADI $W_5$ dimension could be integrated into the SG-SST assessment protocol as an AI-specific module, providing a standardized way to evaluate whether AI adoption is creating or mitigating psychosocial risks. WADI items $W_5.15$--$W_5.18$ (management-reported) are specifically designed to align with SG-SST reporting requirements.

More broadly, the under-investment theorem (Proposition 2) identifies a health externality: firms do not capture the full social benefit of improved $W_5$ because reduced healthcare costs and lower disability rates benefit the public health system, not the firm. This externality justifies public subsidies for workplace redesign aimed at improving psychosocial outcomes---not as a labor rights measure (though it is that too) but as an efficiency correction to a market failure.

% -----------------------------------------------------------------------------
\subsection{Governance for Sustained Human-Centricity}
\label{sec:governance}

The $W_2$ dimension---decision authority allocation---is not only the binding constraint for individual firms but also the critical governance challenge at the policy level. Our model shows that even when human-centric design is profit-maximizing ($H^A / (H^A + H^C) > \theta^*$), firms may under-invest in $W_2$ due to the Athey-Bryan-Gans mechanism: principals resist distributing authority to AI-augmented agents because delegation reduces the agent's incentive to invest in judgment \citep{AtheyBryanGans2020}. This creates a governance paradox: the dimension most critical for augmentation is the dimension where organizational inertia is strongest.

Three governance interventions emerge from the model:

\textbf{Regulatory minimums.} The automation trap can be broken by regulatory mandates that set minimum $W_2$ standards---e.g., requiring that firms using AI for performance-affecting decisions (hiring, task allocation, performance evaluation) must demonstrate that workers have structured appeal and override mechanisms. This is not a prohibition on algorithmic management but a minimum design standard that prevents the worst $W_2$ configurations. The EU AI Act's requirements for human oversight of high-risk AI systems are a step in this direction, but they operate at the technology level rather than the workplace design level that WADI captures.

\textbf{Worker representation in AI governance.} WADI item $W_2.7$ (``The organization has an AI governance body that includes worker representatives'') captures an institutional mechanism that sustains $W_2$ over time. Without worker voice in governance, the tendency toward centralization---driven by managerial risk aversion and the perceived efficiency of algorithmic decisions---will erode initial $W_2$ gains. Colombian law already mandates joint health and safety committees (Comit\'{e}s Paritarios de Seguridad y Salud en el Trabajo, COPASST); extending this model to AI governance would provide institutional infrastructure for sustained human-centricity.

\textbf{Transparency requirements.} Sustained $W_2$ requires sustained $W_1$: workers cannot exercise real authority over AI decisions they cannot understand. Governance frameworks should therefore couple authority distribution mandates with transparency requirements, ensuring that the complementarity between $W_1$ and $W_2$ is maintained at the policy level.

% -----------------------------------------------------------------------------
\subsection{Sustainability and Circular Economy}
\label{sec:sustainability_discuss}

The sustainability constraint (Equation \ref{eq:energy_constraint}) formalizes the tension between augmentation quality and environmental impact. The formal result is a \textit{sustainability-design tradeoff}: when the energy constraint binds, firms must balance the cognitive benefits of more transparent, explainable AI ($W_1$) against the computational cost of providing that transparency.

This tradeoff has three practical dimensions:

\textbf{Energy costs of explanation.} Explainable AI (XAI) techniques---LIME, SHAP, attention visualization, chain-of-thought reasoning---require additional computation beyond the base model inference. \cite{Strubell2019} document that the computational costs of state-of-the-art NLP can be substantial. Under carbon budget constraints ($\bar{E}$), firms face a genuine choice between running a more powerful but opaque AI model (higher $\phi_0(D)$, lower $W_1$) and a less powerful but transparent model (lower $\phi_0(D)$, higher $W_1$). The modified first-order condition for $W_1$ (Appendix A.6) includes a shadow price $\mu > 0$ on the energy constraint, effectively raising the cost of interface improvement.

\textbf{Circular economy for cognitive infrastructure.} The circular economy principle---reduce, reuse, recycle---applies to cognitive infrastructure in three ways: (i) reusable AI models that can be fine-tuned for multiple tasks rather than trained from scratch reduce the energy cost per unit of $D$; (ii) privacy-preserving data reuse enables learning loops ($W_4$) without redundant data collection; and (iii) developing durable human skills ($H^A$ with low depreciation rate $\delta$) reduces the need for continuous retraining. These practices lower the effective cost of both $D$ and $W$, relaxing the sustainability constraint.

\textbf{Sustainability as design dimension.} While our model treats sustainability as a constraint, future extensions could incorporate it as a sixth WADI dimension ($W_6$), measuring the firm's deliberate attention to the environmental impact of its AI deployment. This would extend the framework to capture the full triad of Society 5.0 pillars---human-centricity (WADI $W_1$--$W_5$), sustainability ($W_6$), and resilience (a potential $W_7$)---in a unified measurement framework. We leave this extension for future work.

These seven discussion themes demonstrate that the $\phi(D, W)$ model and the WADI instrument are not isolated academic contributions but provide a coherent framework for the interconnected challenges of the Society 5.0 transition: from individual firm diagnostics to national education policy, from occupational health regulation to sustainability governance. The common thread is that \textit{design choices matter}---and that making them well requires measurement.

\section{Limitations}
\label{sec:limitations}

Several limitations should be acknowledged, each of which defines an avenue for future research.

\textbf{First, the empirical analysis uses proxy variables, not direct WADI measurement.} The EDIT survey's management practice variables capture organizational capacity for structured decision-making but were not designed to measure the five WADI dimensions. The technology investment variable includes non-AI technology, and the innovation typology is an indirect outcome rather than the model's focal dependent variable (labor productivity). These proxies provide directional evidence but cannot substitute for the firm-level WADI survey data required to test the model's six predictions rigorously. The proxy-to-construct mapping is a necessary compromise given that WADI has not yet been deployed.

\textbf{Second, the data are from Colombia.} While the EDIT is a comprehensive census-style survey ($N = 6{,}799$), results may not generalize to economies with different institutional environments, labor market structures, or AI adoption patterns. Colombian manufacturing is characterized by a large share of small and medium enterprises with relatively low digital maturity, which may make the automation trap dynamics more pronounced than in advanced economies where the transition to augmentation is further along. Multi-country replication---particularly in economies at different stages of the S5.0 transition---is essential.

\textbf{Third, the analysis is cross-sectional.} The EDIT X covers 2019--2020, providing a snapshot rather than a trajectory. The automation trap (Proposition 3) is fundamentally a dynamic prediction---firms converge to one of two equilibria over time---and testing it requires panel data that tracks the co-evolution of $W$ and $H^A$ over multiple periods. The 2019--2020 window also precedes the generative AI wave of 2022--2023, meaning that the technology landscape has shifted substantially since the data were collected.

\textbf{Fourth, the WADI instrument has not been field-validated.} We propose WADI as a theory-grounded and literature-validated instrument, but it has not been tested with actual respondents. The 36-item structure, scoring methodology, and proposed validation targets are informed by psychometric best practices and adapted from validated source instruments, but the five-factor structure, internal consistency, test-retest reliability, and predictive validity remain empirical questions that can only be resolved through the deployment protocol described in Section \ref{sec:wadi_validation}. Until field validation is complete, WADI should be regarded as a rigorously constructed proposal rather than a proven measurement tool.

\textbf{Fifth, the model assumes rational firm optimization.} Proposition 1 (human-centricity as optimality) relies on firms choosing $W$ to maximize profits. In practice, organizational inertia, bounded rationality, managerial resistance to distributing authority, and path dependence may prevent firms from reaching their optimal $W^*$ even when the conditions ($H^A / (H^A + H^C) > \theta^*$) are met. The automation trap result partially addresses this---it explains why firms may be stuck at a suboptimal equilibrium---but the model does not formalize the behavioral and political economy barriers that impede transition. Integrating insights from organizational behavior theory into the formal model is a natural extension.

\textbf{Sixth, the systematic review is limited to SCOPUS.} While SCOPUS provides broad coverage across engineering, management, psychology, and economics journals, the restriction to a single database and to the 2020--2026 window means that some relevant work---particularly foundational contributions in organizational economics and older work design literature---enters only through backward citation rather than systematic search. The 120-paper corpus is comprehensive for the recent S5.0/AI workplace literature but does not claim exhaustiveness.

% =============================================================================

\section{Conclusion}
\label{sec:conclusion}

This paper has pursued a single objective: to give operational content to the Society 5.0 promise of human-centric technology integration. We have argued that the gap between aspiration and implementation can be closed by treating workplace design as an endogenous choice variable in the augmentation production function---replacing $\phi(D)$ with $\phi(D, W)$---and by developing a measurement instrument (WADI) that makes this choice observable, diagnosable, and optimizable.

Three results anchor the paper. First, the formal model establishes that human-centricity is not altruism---it is the profit-maximizing strategy when the workforce's augmentable cognitive capital exceeds a critical threshold (Proposition 1). This result resolves the tension between the normative language of the S5.0 literature (``technology should serve human flourishing'') and the economic reality that firms adopt designs that maximize returns: human-centricity and profitability are aligned when $H^A$ is the scarce factor. But human-centricity is not always optimal (Corollary 1), and telling firms to adopt it as a universal prescription is wrong. The optimality depends on workforce composition, design costs, and output prices---variables that WADI measures.

Second, the systematic review of 120 papers reveals a landscape of uneven evidence density that mirrors the uneven attention firms and policymakers give to different design dimensions. Psychosocial outcomes ($W_5$) are extensively studied (87 papers), but the organizational design mechanisms that produce those outcomes are not: decision authority ($W_2$, 14 papers) and especially task orchestration ($W_3$, 4 papers) are dramatically under-researched relative to their theoretical importance. This imbalance is not accidental; it reflects the broader tendency to study consequences (wellbeing, productivity) rather than causes (how work is designed). WADI redirects attention to the design choices that determine outcomes.

Third, the automation trap (Proposition 3) provides the dynamic mechanism that makes the transition problem difficult. Firms with low initial $H^A$ rationally choose automation-centered designs that do not develop augmentable skills, creating path dependence that perpetuates under-investment in human-centricity. Escaping this trap requires a coordinated ``big push'' in both workplace redesign and education investment---a policy bundle that existing interventions, which target technology adoption or skills training in isolation, fail to provide.

Decision authority allocation ($W_2$) emerges as the binding constraint at every level of analysis: it has the thinnest evidence base, the strongest theoretical link to the augmentation-vs.-automation bifurcation, the most direct interaction with other dimensions, and the most significant policy implications. Whether a firm's AI deployment augments or automates depends, more than on any other factor, on whether the workers who interact with the AI have the real authority to use their augmented judgment. This finding has immediate practical relevance: firms seeking to improve their augmentation outcomes should start by redistributing decision authority, not by upgrading interfaces or buying more powerful AI.

We close with three propositions for the research community. First, the Society 5.0 literature must move from conceptual frameworks to measurable constructs; WADI provides one such tool, and others should follow. Second, empirical research on AI at work should control for workplace design---omitting $W$ from regressions of productivity on $D$ biases the estimated effect of technology, which may explain the heterogeneous and sometimes contradictory findings in the emerging literature. Third, policy for the S5.0 transition should bundle technology subsidies, education investment, and workplace design standards, recognizing that the complementarity between these interventions means that each is less effective without the others.

The transition from automation to augmentation is not inevitable, and it is not primarily about technology. It is about design---and design is a choice that can be measured, optimized, and governed.

% ============================================================
\bibliographystyle{plainnat}
\bibliography{references}

@techreport{EUCommission2021,
  author      = {{European Commission}},
  title       = {Industry 5.0: Towards a Sustainable, Human-Centric and Resilient European Industry},
  institution = {European Commission, Directorate-General for Research and Innovation},
  year        = {2021},
  type        = {R\&I Paper Series}
}

@techreport{BrequeDeNulPetridis2021,
  author      = {Breque, Maija and De Nul, Lars and Petridis, Athanasios},
  title       = {Industry 5.0},
  institution = {European Commission},
  year        = {2021},
  type        = {R\&I Paper Series}
}

@techreport{JapanCabinetOffice2016,
  author      = {{Japan Cabinet Office}},
  title       = {Society 5.0: 5th Science and Technology Basic Plan},
  institution = {Cabinet Office, Government of Japan},
  year        = {2016}
}

@article{Fukuyama2018,
  author  = {Fukuyama, Mayumi},
  title   = {Society 5.0: Aiming for a New Human-Centered Society},
  journal = {Japan SPOTLIGHT},
  year    = {2018},
  volume  = {27},
  pages   = {47--50}
}

@article{Xu2021,
  author  = {Xu, Xun and Lu, Yuqian and Vogel-Heuser, Birgit and Wang, Lihui},
  title   = {Industry 4.0 and Industry 5.0: Inception, Conception and Perception},
  journal = {Journal of Manufacturing Systems},
  year    = {2021},
  volume  = {61},
  pages   = {530--535},
  doi     = {10.1016/j.jmsy.2021.10.006}
}

@article{Maddikunta2022,
  author  = {Maddikunta, Praveen Kumar Reddy and Pham, Quoc-Viet and Prabadevi, B. and Deepa, N. and Dev, Kapal and Gadekallu, Thippa Reddy and Ruby, Rukhsana and Liyanage, Madhusanka},
  title   = {Industry 5.0: A Survey on Enabling Technologies and Potential Applications},
  journal = {Journal of Industrial Information Integration},
  year    = {2022},
  volume  = {26},
  pages   = {100257},
  doi     = {10.1016/j.jii.2021.100257}
}

@article{Leng2022,
  author  = {Leng, Jiewu and Sha, Weinan and Wang, Baicun and Zheng, Pai and Zhuang, Cunbo and Liu, Qiang and Wuest, Thorsten and Mourtzis, Dimitris and Wang, Lihui},
  title   = {Industry 5.0: Prospect and Retrospect},
  journal = {Journal of Manufacturing Systems},
  year    = {2022},
  volume  = {65},
  pages   = {279--295},
  doi     = {10.1016/j.jmsy.2022.09.017}
}

@article{Grabowska2022,
  author  = {Grabowska, Sandra and Saniuk, Sebastian and Gajdzik, Bozena},
  title   = {Industry 5.0: Improving Humanization and Sustainability of {Industry 4.0}},
  journal = {Scientometrics},
  year    = {2022},
  volume  = {127},
  number  = {6},
  pages   = {3117--3144},
  doi     = {10.1007/s11192-022-04370-1}
}

@article{Ivanov2023,
  author  = {Ivanov, Dmitry},
  title   = {The {Industry 5.0} Framework: Viability-Based Integration of the Value and Resilience of Supply Chains},
  journal = {International Journal of Production Research},
  year    = {2023},
  volume  = {61},
  number  = {21},
  pages   = {7162--7178},
  doi     = {10.1080/00207543.2022.2118892}
}

@inproceedings{Amershi2019,
  author    = {Amershi, Saleema and Weld, Dan and Vorvoreanu, Mihaela and Fourney, Adam and Nushi, Besmira and Collisson, Penny and Suh, Jina and Iqbal, Shamsi and Bennett, Paul N. and Inkpen, Kori and Teevan, Jaime and Kiber, Ruth and Horvitz, Eric},
  title     = {Guidelines for Human-{AI} Interaction},
  booktitle = {Proceedings of the 2019 CHI Conference on Human Factors in Computing Systems},
  year      = {2019},
  pages     = {1--13},
  doi       = {10.1145/3290605.3300233}
}

@book{Shneiderman2022,
  author    = {Shneiderman, Ben},
  title     = {Human-Centered {AI}},
  publisher = {Oxford University Press},
  year      = {2022}
}

@article{Endsley2017,
  author  = {Endsley, Mica R.},
  title   = {From Here to Autonomy: Lessons Learned from Human--Automation Research},
  journal = {Human Factors},
  year    = {2017},
  volume  = {59},
  number  = {1},
  pages   = {5--27},
  doi     = {10.1177/0018720816681350}
}

@article{Endsley1995,
  author  = {Endsley, Mica R.},
  title   = {Toward a Theory of Situation Awareness in Dynamic Systems},
  journal = {Human Factors},
  year    = {1995},
  volume  = {37},
  number  = {1},
  pages   = {32--64},
  doi     = {10.1518/001872095779049543}
}

@techreport{BrynjolfssonLiRaymond2023,
  author      = {Brynjolfsson, Erik and Li, Danielle and Raymond, Lindsey R.},
  title       = {Generative {AI} at Work},
  institution = {National Bureau of Economic Research},
  year        = {2023},
  type        = {Working Paper},
  number      = {31161},
  doi         = {10.3386/w31161}
}

@techreport{DellAcqua2023,
  author      = {Dell'Acqua, Fabrizio and McFowland III, Edward and Mollick, Ethan R. and Lifshitz-Assaf, Hila and Kellogg, Katherine and Rajendran, Saran and Krayer, Lisa and Candelon, Fran{\c{c}}ois and Lakhani, Karim R.},
  title       = {Navigating the Jagged Technological Frontier: Field Experimental Evidence of the Effects of {AI} on Knowledge Worker Productivity and Quality},
  institution = {Harvard Business School},
  year        = {2023},
  type        = {Working Paper},
  number      = {24-013}
}

@article{NoyZhang2023,
  author  = {Noy, Shakked and Zhang, Whitney},
  title   = {Experimental Evidence on the Productivity Effects of Generative Artificial Intelligence},
  journal = {Science},
  year    = {2023},
  volume  = {381},
  number  = {6654},
  pages   = {187--192},
  doi     = {10.1126/science.adh2586}
}

@article{AghionTirole1997,
  author  = {Aghion, Philippe and Tirole, Jean},
  title   = {Formal and Real Authority in Organizations},
  journal = {Journal of Political Economy},
  year    = {1997},
  volume  = {105},
  number  = {1},
  pages   = {1--29},
  doi     = {10.1086/262063}
}

@article{Garicano2000,
  author  = {Garicano, Luis},
  title   = {Hierarchies and the Organization of Knowledge in Production},
  journal = {Journal of Political Economy},
  year    = {2000},
  volume  = {108},
  number  = {5},
  pages   = {874--904},
  doi     = {10.1086/317671}
}

@article{BloomVanReenen2007,
  author  = {Bloom, Nicholas and Van Reenen, John},
  title   = {Measuring and Explaining Management Practices across Firms and Countries},
  journal = {The Quarterly Journal of Economics},
  year    = {2007},
  volume  = {122},
  number  = {4},
  pages   = {1351--1408},
  doi     = {10.1162/qjec.2007.122.4.1351}
}

@article{Bloom2012,
  author  = {Bloom, Nicholas and Sadun, Raffaella and Van Reenen, John},
  title   = {The Organization of Firms across Countries},
  journal = {The Quarterly Journal of Economics},
  year    = {2012},
  volume  = {127},
  number  = {4},
  pages   = {1663--1705},
  doi     = {10.1093/qje/qje029}
}

@article{Bloom2014,
  author  = {Bloom, Nicholas and Lemos, Renata and Sadun, Raffaella and Scur, Daniela and Van Reenen, John},
  title   = {The New Empirical Economics of Management},
  journal = {Journal of the European Economic Association},
  year    = {2014},
  volume  = {12},
  number  = {4},
  pages   = {835--876},
  doi     = {10.1111/jeea.12094}
}

@article{Karasek1979,
  author  = {Karasek, Robert A.},
  title   = {Job Demands, Job Decision Latitude, and Mental Strain: Implications for Job Redesign},
  journal = {Administrative Science Quarterly},
  year    = {1979},
  volume  = {24},
  number  = {2},
  pages   = {285--308},
  doi     = {10.2307/2392498}
}

@article{HackmanOldham1976,
  author  = {Hackman, J. Richard and Oldham, Greg R.},
  title   = {Motivation through the Design of Work: Test of a Theory},
  journal = {Organizational Behavior and Human Performance},
  year    = {1976},
  volume  = {16},
  number  = {2},
  pages   = {250--279},
  doi     = {10.1016/0030-5073(76)90016-7}
}

@article{ParkerGrote2022,
  author  = {Parker, Sharon K. and Grote, Gudela},
  title   = {Automation, Algorithms, and Beyond: Why Work Design Matters More Than Ever in a Digital World},
  journal = {Applied Psychology},
  year    = {2022},
  volume  = {71},
  number  = {4},
  pages   = {1171--1204},
  doi     = {10.1111/apps.12241}
}

@article{Demerouti2001,
  author  = {Demerouti, Evangelia and Bakker, Arnold B. and Nachreiner, Friedhelm and Schaufeli, Wilmar B.},
  title   = {The Job Demands--Resources Model of Burnout},
  journal = {Journal of Applied Psychology},
  year    = {2001},
  volume  = {86},
  number  = {3},
  pages   = {499--512},
  doi     = {10.1037/0021-9010.86.3.499}
}

@article{BakkerDemerouti2017,
  author  = {Bakker, Arnold B. and Demerouti, Evangelia},
  title   = {Job Demands--Resources Theory: Taking Stock and Looking Forward},
  journal = {Journal of Occupational Health Psychology},
  year    = {2017},
  volume  = {22},
  number  = {3},
  pages   = {273--285},
  doi     = {10.1037/ocp0000056}
}

@article{AutorLevyMurnane2003,
  author  = {Autor, David H. and Levy, Frank and Murnane, Richard J.},
  title   = {The Skill Content of Recent Technological Change: An Empirical Exploration},
  journal = {The Quarterly Journal of Economics},
  year    = {2003},
  volume  = {118},
  number  = {4},
  pages   = {1279--1333},
  doi     = {10.1162/003355303322552801}
}

@article{AcemogluRestrepo2022,
  author  = {Acemoglu, Daron and Restrepo, Pascual},
  title   = {Tasks, Automation, and the Rise in {U.S.} Wage Inequality},
  journal = {Econometrica},
  year    = {2022},
  volume  = {90},
  number  = {5},
  pages   = {1973--2016},
  doi     = {10.3982/ECTA19815}
}

@techreport{Espinal2026a,
  author      = {Espinal, Cristian},
  title       = {Augmented Human Capital: Cognitive Factor Decomposition in {AI}-Augmented Economies},
  institution = {INPLUX},
  year        = {2026},
  type        = {CFE Paper 1}
}

@techreport{Espinal2026b,
  author      = {Espinal, Cristian},
  title       = {Cognitive Factor Economics: A Research Program},
  institution = {INPLUX},
  year        = {2026},
  type        = {Working Paper}
}

@article{Leng2023,
  author  = {Leng, Jiewu and Zhong, Yuxuan and Lin, Zisheng and others},
  title   = {Towards a Digital Circular Economy and {Industry 5.0}},
  journal = {Journal of Manufacturing Systems},
  year    = {2023},
  volume  = {68},
  pages   = {442--459},
  doi     = {10.1016/j.jmsy.2023.05.013}
}

@article{DeSousaJabbour2023,
  author  = {{De Sousa Jabbour}, Ana Beatriz Lopes and Jabbour, Charbel Jose Chiappetta and Godinho Filho, Moacir and Roubaud, David},
  title   = {Industry 5.0 and the Circular Economy: Leveraging Sustainability through Digital Transformation},
  journal = {International Journal of Production Research},
  year    = {2023},
  volume  = {61},
  doi     = {10.1080/00207543.2023.2217812}
}

@inproceedings{Strubell2019,
  author    = {Strubell, Emma and Ganesh, Ananya and McCallum, Andrew},
  title     = {Energy and Policy Considerations for Deep Learning in {NLP}},
  booktitle = {Proceedings of the 57th Annual Meeting of the Association for Computational Linguistics},
  year      = {2019},
  pages     = {3645--3650},
  doi       = {10.18653/v1/P19-1355}
}

@article{MasoodEgger2020,
  author  = {Masood, Tariq and Egger, Johannes},
  title   = {Augmented Reality in Support of {Industry 4.0}: Implementation Challenges and Success Factors},
  journal = {Robotics and Computer-Integrated Manufacturing},
  year    = {2020},
  volume  = {58},
  pages   = {181--195},
  doi     = {10.1016/j.rcim.2019.02.003}
}

@article{Teece2018,
  author  = {Teece, David J.},
  title   = {Business Models and Dynamic Capabilities},
  journal = {Long Range Planning},
  year    = {2018},
  volume  = {51},
  number  = {1},
  pages   = {40--49},
  doi     = {10.1016/j.lrp.2017.06.007},
  note    = {Also published as: Dynamic Capabilities as (Workable) Management Systems Theory, JIBS 49(7)}
}

@book{ArgyrisSchon1978,
  author    = {Argyris, Chris and Sch{\"o}n, Donald A.},
  title     = {Organizational Learning: A Theory of Action Perspective},
  publisher = {Addison-Wesley},
  year      = {1978}
}

@article{AtheyBryanGans2020,
  author  = {Athey, Susan and Bryan, Kevin and Gans, Joshua S.},
  title   = {The Allocation of Decision Authority to Human and Artificial Intelligence},
  journal = {AEA Papers and Proceedings},
  year    = {2020},
  volume  = {110},
  pages   = {80--84},
  doi     = {10.1257/pandp.20201034}
}

@article{ParkWooKim2024,
  author  = {Park, Jungkun and Woo, Sang Eun and Kim, Jiin},
  title   = {The {AAAW} Scale: Attitudes towards Artificial Intelligence at Work},
  journal = {Journal of Occupational and Organizational Psychology},
  year    = {2024},
  volume  = {97},
  number  = {3},
  note    = {25 items, 6 dimensions, N=2,841}
}

@article{Colledani2025,
  author  = {Colledani, Davide},
  title   = {Measuring Technostress in Everyday Life: Development and Validation of a Short Scale},
  journal = {Human Behavior and Emerging Technologies},
  year    = {2025},
  note    = {16 items, 4 factors}
}

@article{Parker2024,
  author  = {Parker, Sharon K.},
  title   = {The {SMART} Model of Work Design: A Higher Order Structure to Help See the Wood from the Trees},
  journal = {Human Resource Management},
  year    = {2024},
  note    = {Wiley. 5 higher-order factors, 20+ characteristics}
}

@article{Ghobakhloo2022SPC,
  author  = {Ghobakhloo, Morteza and Fathi, Masood and Iranmanesh, Mohammad and Maroufkhani, Parisa and Morales, Mario E.},
  title   = {Industry 5.0 Implications for Inclusive Sustainable Manufacturing: A Systematic Review and Emerging Themes},
  journal = {Sustainable Production and Consumption},
  year    = {2022},
  volume  = {33},
  pages   = {1--20},
  doi     = {10.1016/j.spc.2022.06.017},
  note    = {11 enablers, fuzzy DEMATEL}
}

@article{GhobakhlooIranmanesh2023,
  author  = {Ghobakhloo, Morteza and Iranmanesh, Mohammad},
  title   = {Digital Transformation Success under {Industry 5.0}: A Need for Digitally Capable People and Sustainable Practices},
  journal = {Corporate Social Responsibility and Environmental Management},
  year    = {2023},
  volume  = {30},
  number  = {3},
  pages   = {964--985},
  doi     = {10.1002/csr.2397}
}

@article{Bag2021CE,
  author  = {Bag, Surajit and Pretorius, Jan Ham C. and Gupta, Shivam and Dwivedi, Yogesh K.},
  title   = {Role of Institutional Pressures and Resources in the Adoption of Big Data Analytics Powered Artificial Intelligence, Sustainable Manufacturing Practices and Circular Economy Capabilities},
  journal = {Technological Forecasting and Social Change},
  year    = {2021},
  volume  = {163},
  pages   = {120420},
  doi     = {10.1016/j.techfore.2020.120420},
  note    = {Validated survey instrument}
}

@article{Buchner2022,
  author  = {Buchner, Josef and Buntins, Katja and Kerres, Michael},
  title   = {The Impact of Augmented Reality on Cognitive Load and Performance: A Systematic Review},
  journal = {Journal of Computer Assisted Learning},
  year    = {2022},
  volume  = {38},
  number  = {3},
  pages   = {285--303},
  doi     = {10.1111/jcal.12617}
}

@article{LongoPadovanoUmbrello2020,
  author  = {Longo, Francesco and Padovano, Antonio and Umbrello, Steven},
  title   = {Value-Oriented and Ethical Technology Engineering in {Industry 5.0}: A Human-Centric Perspective for the Design of the Factory of the Future},
  journal = {Applied Sciences},
  year    = {2020},
  volume  = {10},
  number  = {12},
  pages   = {4182},
  doi     = {10.3390/app10124182}
}

@article{Nahavandi2019,
  author  = {Nahavandi, Saeid},
  title   = {Industry 5.0---A Human-Centric Solution},
  journal = {Sustainability},
  year    = {2019},
  volume  = {11},
  number  = {16},
  pages   = {4371},
  doi     = {10.3390/su11164371}
}

@article{Berx2025,
  author = {Berx, N. and Decr\'{e}, W. and De Schutter, J. and Pintelon, L.},
  title = {A harmonious synergy between robotic performance and well-being in human-robot collaboration: A vision and key recommendations},
  journal = {Annual Reviews in Control},
  year = {2025},
  doi = {10.1016/j.arcontrol.2024.100984}
}

@article{ChigbuMakapela2025,
  author = {Chigbu, B. I. and Makapela, S. L.},
  title = {{AI} in education, sustainability, and the future of work: An integrative review of industry 5.0, education 5.0, and work 5.0},
  journal = {Journal of Open Innovation: Technology, Market, and Complexity},
  year = {2025},
  doi = {10.1016/j.joitmc.2025.100645}
}

@article{Dang2025,
  author = {Dang, J.-F. and Chen, T.-L. and Huang, H.-Y.},
  title = {The human-centric framework integrating knowledge distillation architecture with fine-tuning mechanism for equipment health monitoring},
  journal = {Advanced Engineering Informatics},
  year = {2025},
  doi = {10.1016/j.aei.2025.103167}
}

@article{Dave2026,
  author = {Dave, B. and Martin, P. and David, S. S. and Kumar, S. and Chakraborty, T.},
  title = {Enhancing healthcare worker mental health via artificial intelligence-driven work process improvement},
  journal = {International Journal of Medical Informatics},
  year = {2026},
  doi = {10.1016/j.ijmedinf.2026.105850}
}

@inproceedings{FernandezDornellesAyala2026,
  author = {Fernandez, G. and de Assis Dornelles, J. and Ayala, N. F.},
  title = {How Industry 4.0 Technologies Are Evolving to Industry 5.0 --- A Systematic Literature Review},
  booktitle = {IFIP Advances in Information and Communication Technology},
  year = {2026},
  doi = {10.1007/978-3-032-03515-8_4}
}

@inproceedings{Sileo2025,
  author = {Sileo, M. and Carriero, G. and Brancato, C. and Calzone, N. and Pierri, F. and Caccavale, F.},
  title = {Safety Features for {HumanTIX}: An Augmented Reality Platform to Enhance Human--Robot Collaboration},
  booktitle = {Lecture Notes in Mechanical Engineering},
  year = {2025},
  doi = {10.1007/978-3-031-72829-7_45}
}

@inproceedings{Zhou2025,
  author = {Zhou, T. and Wan, Y. and Liu, Y. and Kumar, M.},
  title = {Enabling Interactive {AI} in Industry 5.0 with {RAG}-Enhanced {GenAI} Chatbots},
  booktitle = {Proceedings of the 31st ICE IEEE/ITMC Conference},
  year = {2025},
  doi = {10.1109/ICE/ITMC65658.2025.11106634}
}

% ============================================================
\appendix
\section{Formal Proofs}
\label{sec:appendix_proofs}

\subsection*{A.1 Setup and Notation}

Consider a firm $f$ that produces output $Y_f$ using the following production technology (inherited from the CFE framework, Espinal 2026):

\begin{equation}
Y_f = F\left(K_f,\; L_f^P H_f^P + \kappa K_f^{Rob},\; L_f^C H_f^C + \phi(D_f, W_f) \cdot L_f^A H_f^A \cdot D_f\right)
\label{eq:production}
\end{equation}

\noindent where $F(\cdot)$ is a constant-returns-to-scale production function satisfying standard Inada conditions, and:

\begin{itemize}
    \item $K_f$ = physical capital (hardware)
    \item $L_f^j$ = labor employed in category $j \in \{P, C, A\}$ (physical, routine-cognitive, augmentable-cognitive)
    \item $H_f^j$ = average human capital of type $j$ in the firm's workforce
    \item $K_f^{Rob}$ = robotic capital, substitutable with physical labor at rate $\kappa$
    \item $D_f$ = digital labor (stock of AI systems)
    \item $W_f = (W_{f1}, \ldots, W_{f5}) \in \mathbb{R}_+^5$ = workplace design vector
\end{itemize}

\noindent For notational convenience, define the three sectors of production as:

\begin{align}
    Z_1 &\equiv K_f \quad \text{(capital sector)} \\
    Z_2 &\equiv L_f^P H_f^P + \kappa K_f^{Rob} \quad \text{(hardware/physical sector)} \\
    Z_3 &\equiv L_f^C H_f^C + \phi(D_f, W_f) \cdot L_f^A H_f^A \cdot D_f \quad \text{(software/cognitive sector)}
\end{align}

\noindent so that $Y_f = F(Z_1, Z_2, Z_3)$.

\subsubsection*{The Augmentation Function}

The augmentation function decomposes multiplicatively:

\begin{equation}
\phi(D, W) = \phi_0(D) \cdot g(W, H^A)
\label{eq:phi_decomp}
\end{equation}

\noindent where:

\begin{assumption}[Technology-Only Augmentation]
\label{ass:phi0}
$\phi_0: \mathbb{R}_+ \to [1, \bar{\phi}_0]$ is twice continuously differentiable with:
\begin{enumerate}[label=(\roman*)]
    \item $\phi_0(0) = 1$ (no AI $\Rightarrow$ no augmentation)
    \item $\phi_0'(D) > 0$ for all $D > 0$ (more AI $\Rightarrow$ more augmentation potential)
    \item $\phi_0''(D) < 0$ for all $D > 0$ (diminishing returns)
    \item $\lim_{D \to \infty} \phi_0(D) = \bar{\phi}_0 < \infty$ (bounded above)
\end{enumerate}
\end{assumption}

\begin{assumption}[Design Multiplier]
\label{ass:g}
$g: \mathbb{R}_+^5 \times \mathbb{R}_+ \to (0, \bar{g}]$ is twice continuously differentiable with:
\begin{enumerate}[label=(\roman*)]
    \item $g(W^{min}, H^A) = g_0 < 1$ for some minimal design $W^{min}$ (poor design dampens augmentation)
    \item $g(W^{auto}, H^A) = 1$ for a neutral design $W^{auto}$ (automation-centered, neither helps nor hinders)
    \item $\frac{\partial g}{\partial W_k} > 0$ for all $k \in \{1, \ldots, 5\}$ and all $(W, H^A) \gg 0$ (improving any design dimension weakly increases augmentation)
    \item $\frac{\partial^2 g}{\partial W_k^2} < 0$ for all $k$ (diminishing returns within each dimension)
    \item $\frac{\partial^2 g}{\partial W_k \partial H^A} > 0$ for all $k$ (\textbf{design-composition complementarity})
\end{enumerate}
\end{assumption}

\begin{assumption}[Design Costs]
\label{ass:costs}
The cost of workplace design $c_W: \mathbb{R}_+^5 \to \mathbb{R}_+$ is separable and convex:
$$c_W(W) = \sum_{k=1}^{5} c_k(W_k), \quad c_k'(W_k) > 0, \quad c_k''(W_k) > 0, \quad c_k(0) = 0$$
\end{assumption}

% =============================================================================
\subsection*{A.2 Proposition 1: Human-Centricity as Optimality}
% =============================================================================

\begin{proposition}[Human-Centricity as Optimality]
\label{prop:hc_optimal}
Under Assumptions \ref{ass:phi0}--\ref{ass:costs}, there exists a threshold $\theta^* > 0$ such that if the workforce's augmentable capital share satisfies $H^A / (H^A + H^C) > \theta^*$, then the profit-maximizing workplace design $W^*$ is human-centric: $W^* > W^{auto}$ (component-wise).
\end{proposition}

\begin{proof}
The firm solves:
\begin{equation}
\max_{D, W, L^A, L^C, L^P} \quad P \cdot F(Z_1, Z_2, Z_3) - c_D \cdot D - c_W(W) - w_A L^A - w_C L^C - w_P L^P
\label{eq:firm_problem}
\end{equation}

\noindent where $P$ is the output price, $c_D$ is the per-unit cost of AI capital, and $w_j$ are wages for each labor type. Taking the first-order condition with respect to $W_k$ for any $k \in \{1, \ldots, 5\}$:

\begin{equation}
\frac{\partial \Pi}{\partial W_k} = P \cdot F_3(Z_1, Z_2, Z_3) \cdot \frac{\partial Z_3}{\partial W_k} - c_k'(W_k) = 0
\label{eq:foc_wk}
\end{equation}

\noindent where $F_3 \equiv \partial F / \partial Z_3 > 0$ (positive marginal product of the cognitive sector). Computing $\partial Z_3 / \partial W_k$:

\begin{equation}
\frac{\partial Z_3}{\partial W_k} = \phi_0(D) \cdot \frac{\partial g(W, H^A)}{\partial W_k} \cdot L^A H^A \cdot D
\label{eq:dz3_dwk}
\end{equation}

\noindent Substituting into \eqref{eq:foc_wk}:

\begin{equation}
P \cdot F_3 \cdot \phi_0(D) \cdot \frac{\partial g}{\partial W_k} \cdot L^A H^A D = c_k'(W_k)
\label{eq:foc_expanded}
\end{equation}

\noindent The left-hand side (LHS) is the marginal revenue product of improving design dimension $k$. Observe that the LHS is:
\begin{itemize}
    \item \textbf{Increasing in $H^A$} (directly, and through $\partial g / \partial W_k$ which is increasing in $H^A$ by Assumption \ref{ass:g}(v))
    \item \textbf{Increasing in $D$} (directly, and through $\phi_0(D)$)
    \item \textbf{Increasing in $L^A$} (directly)
\end{itemize}

\noindent The right-hand side (RHS) $c_k'(W_k)$ is increasing in $W_k$ (convex costs).

\medskip

\noindent \textbf{At the neutral design $W_k = W_k^{auto}$}, we have $g(W^{auto}, H^A) = 1$ by Assumption \ref{ass:g}(ii). The firm will choose $W_k^* > W_k^{auto}$ if and only if:

\begin{equation}
\underbrace{P \cdot F_3 \cdot \phi_0(D) \cdot \frac{\partial g}{\partial W_k}\bigg|_{W=W^{auto}} \cdot L^A H^A D}_{\text{Marginal benefit of improving design beyond neutral}} > \underbrace{c_k'(W_k^{auto})}_{\text{Marginal cost at neutral}}
\label{eq:hc_condition}
\end{equation}

\noindent By design-composition complementarity (Assumption \ref{ass:g}(v)), $\partial g / \partial W_k$ is increasing in $H^A$. Therefore, the LHS of \eqref{eq:hc_condition} is increasing in $H^A$.

\medskip

\noindent For $H^A \to 0$: $L^A H^A \to 0$ (no augmentable workers), so LHS $\to 0 < c_k'(W_k^{auto})$ = RHS. The firm sets $W_k^* < W_k^{auto}$ (automation-centered).

\medskip

\noindent For $H^A \to \infty$: LHS $\to \infty > c_k'(W_k^{auto})$ = RHS. The firm sets $W_k^* > W_k^{auto}$ (human-centric).

\medskip

\noindent By the intermediate value theorem, there exists a threshold $\tilde{H}^A_k > 0$ such that:

$$W_k^* \geq W_k^{auto} \iff H^A \geq \tilde{H}^A_k$$

\noindent Define the share threshold as $\theta^*_k \equiv \tilde{H}^A_k / (\tilde{H}^A_k + H^C)$, and $\theta^* \equiv \max_k \theta^*_k$. Then for $H^A / (H^A + H^C) > \theta^*$, we have $W_k^* > W_k^{auto}$ for all $k \in \{1, \ldots, 5\}$, i.e., $W^* > W^{auto}$ component-wise. \qed
\end{proof}

\begin{corollary}[Human-Centricity is NOT Always Optimal]
\label{cor:not_always}
When $H^A / (H^A + H^C) < \min_k \theta^*_k$, the profit-maximizing design satisfies $W_k^* < W_k^{auto}$ for all $k$. The firm rationally chooses automation-centered design.
\end{corollary}

\begin{remark}
The threshold $\theta^*$ depends on the design cost structure $c_k(\cdot)$, the level of AI capital $D$, output price $P$, and the shape of $g(\cdot)$. Sectors with low design costs (e.g., knowledge services where organizational change is cheaper than in manufacturing) will have lower $\theta^*$, making human-centric design optimal for a wider range of workforce compositions.
\end{remark}

% =============================================================================
\subsection*{A.3 Proposition 2: Under-Investment in Human-Centricity}
% =============================================================================

\begin{proposition}[Under-Investment Theorem]
\label{prop:underinvest_formal}
The privately optimal workplace design $W^{priv}$ is strictly below the socially optimal design $W^{soc}$ when the following externalities are present:
\begin{enumerate}[label=(\alph*)]
    \item \textbf{Labor mobility externality:} workers trained through augmented work carry $\Delta H^A$ to other firms
    \item \textbf{Knowledge spillover externality:} workplace design innovations diffuse to competitor firms
    \item \textbf{Health externality:} improved $W_5$ reduces public healthcare costs not borne by the firm
\end{enumerate}
\end{proposition}

\begin{proof}

\noindent \textbf{Step 1: The social planner's problem.}

The social planner maximizes total surplus across all firms $f \in \{1, \ldots, N\}$:

\begin{equation}
\max_{\{W_f\}} \sum_{f=1}^{N} \left[ P \cdot F_f(\cdot | \phi(D_f, W_f)) - c_W(W_f) \right] + \underbrace{\sum_{f=1}^{N} E_f(W_f)}_{\text{externalities}}
\end{equation}

\noindent where $E_f(W_f)$ captures the three externalities.

\medskip

\noindent \textbf{Step 2: Characterizing each externality.}

\medskip

\noindent \textit{(a) Labor mobility.} When firm $f$ invests in $W_f$ (especially $W_{f4}$, learning loops), its workers accumulate additional augmentable capital $\Delta H^A_{if} = h(W_{f4}) > 0$ with $h' > 0$. A fraction $\lambda \in (0,1)$ of workers leave per period, carrying $\Delta H^A$ to other firms. The private benefit to firm $f$ is:

$$B_f^{priv} = (1 - \lambda) \cdot V(\Delta H^A_{if})$$

\noindent but the social benefit is:

$$B_f^{soc} = V(\Delta H^A_{if})$$

\noindent since the economy benefits from the full $\Delta H^A$ regardless of where the worker is employed. The wedge is $\lambda \cdot V(\Delta H^A_{if}) > 0$, which the firm ignores when choosing $W_{f4}$.

\medskip

\noindent \textit{(b) Knowledge spillover.} Design innovations at firm $f$ are partially observed and adopted by $M$ competitor firms, each of which gains $\sigma \cdot g(W_f)$ where $\sigma \in (0,1)$ is the spillover rate. The firm ignores the benefit $M \cdot \sigma \cdot g(W_f)$ when choosing $W_f$.

\medskip

\noindent \textit{(c) Health externality.} Improved psychosocial environment ($W_{f5}$) reduces worker stress, burnout, and related healthcare utilization. Let $\mathcal{H}(W_{f5})$ be the reduction in public healthcare costs per worker, with $\mathcal{H}' > 0$. The firm bears only a fraction $\tau$ of these costs (through insurance contributions); the public system bears $(1-\tau)$. The wedge is $(1-\tau) \cdot L_f \cdot \mathcal{H}(W_{f5}) > 0$.

\medskip

\noindent \textbf{Step 3: Comparing FOCs.}

The private FOC for $W_k$ is (from equation \ref{eq:foc_expanded}):

$$P \cdot F_3 \cdot \phi_0 \cdot g_k \cdot L^A H^A D = c_k'(W_k)$$

The social FOC adds the marginal externality:

$$P \cdot F_3 \cdot \phi_0 \cdot g_k \cdot L^A H^A D + \underbrace{\frac{\partial E}{\partial W_k}}_{> 0} = c_k'(W_k)$$

Since $\partial E / \partial W_k > 0$ for all three externalities and $c_k''(W_k) > 0$, the social optimum requires a higher $W_k$ to satisfy the FOC:

$$W_k^{soc} > W_k^{priv} \quad \text{for all } k$$

\noindent The under-investment is largest for dimensions with the largest externalities:
\begin{itemize}
    \item $W_4$ (learning loops): highest labor mobility externality (workers carry learned skills)
    \item $W_5$ (psychosocial): highest health externality
    \item $W_1, W_2, W_3$: primarily knowledge spillover externality
\end{itemize}
\qed
\end{proof}

\begin{remark}[Policy Implication]
The under-investment theorem provides a theoretical justification for public subsidies to human-centric workplace design. The optimal subsidy per dimension $k$ equals the marginal externality $\partial E / \partial W_k$ evaluated at the social optimum. This is not a redistributive policy (transferring from firms to workers) but an efficiency correction.
\end{remark}

% =============================================================================
\subsection*{A.4 Proposition 3: Path Dependence and the Automation Trap}
% =============================================================================

\begin{proposition}[Automation Trap]
\label{prop:path_dep}
The dynamic system $(W_f(t), H^A_f(t))$ exhibits two locally stable equilibria:
\begin{enumerate}[label=(\roman*)]
    \item \textbf{Low equilibrium (Automation Trap):} $(W^L, H^{A,L})$ with $W^L < W^{auto}$ and $H^{A,L}$ small
    \item \textbf{High equilibrium (Augmentation Regime):} $(W^H, H^{A,H})$ with $W^H > W^{auto}$ and $H^{A,H}$ large
\end{enumerate}
There exists an unstable interior equilibrium $(W^U, H^{A,U})$ that separates the basins of attraction. Initial conditions determine which equilibrium the firm converges to.
\end{proposition}

\begin{proof}

\noindent \textbf{Step 1: The dynamic system.}

\noindent Assume the firm adjusts workplace design toward its current optimum at speed $\alpha > 0$, and augmentable human capital evolves through accumulation and depreciation:

\begin{align}
\dot{W}_f &= \alpha \left[ W^*(H^A_f) - W_f \right] \label{eq:dyn_W} \\
\dot{H}^A_f &= \beta(W_f) \cdot I_{edu} - \delta(W_f) \cdot H^A_f \label{eq:dyn_HA}
\end{align}

\noindent where:
\begin{itemize}
    \item $W^*(H^A)$ is the optimal design from Proposition \ref{prop:hc_optimal}, with $\frac{dW^*}{dH^A} > 0$ (from the FOC comparative static)
    \item $\beta(W) > 0$ is the $H^A$ formation rate, with $\beta'(W) > 0$ (better workplace design develops more augmentable capital through learning loops $W_4$ and task exposure $W_3$)
    \item $\delta(W) > 0$ is the $H^A$ depreciation rate, with $\delta'(W) < 0$ (better workplace design slows skill obsolescence by keeping workers engaged with frontier tasks)
    \item $I_{edu} > 0$ is exogenous educational investment
\end{itemize}

\medskip

\noindent \textbf{Step 2: Steady states.}

\noindent At a steady state $(\bar{W}, \bar{H}^A)$, both $\dot{W} = 0$ and $\dot{H}^A = 0$:

\begin{align}
\bar{W} &= W^*(\bar{H}^A) \quad \text{(design is optimal given $H^A$)} \label{eq:ss1} \\
\bar{H}^A &= \frac{\beta(\bar{W}) \cdot I_{edu}}{\delta(\bar{W})} \quad \text{($H^A$ is in accumulation-depreciation balance)} \label{eq:ss2}
\end{align}

\noindent Substituting \eqref{eq:ss1} into \eqref{eq:ss2}, define:

\begin{equation}
\Psi(H^A) \equiv \frac{\beta\big(W^*(H^A)\big) \cdot I_{edu}}{\delta\big(W^*(H^A)\big)} - H^A
\label{eq:psi}
\end{equation}

\noindent Steady states satisfy $\Psi(H^A) = 0$. We show that $\Psi$ has (at least) three zeros.

\medskip

\noindent \textbf{Step 3: Shape of $\Psi$.}

\begin{itemize}
    \item \textit{At $H^A = 0$:} $W^*(0) = W^{min}$ (from Corollary \ref{cor:not_always}), so $\Psi(0) = \frac{\beta(W^{min}) \cdot I_{edu}}{\delta(W^{min})} > 0$. There exists a small positive steady state $H^{A,L}$ near zero.

    \item \textit{Derivative of $\Psi$:}
    \begin{equation}
    \Psi'(H^A) = \frac{d}{dH^A}\left[\frac{\beta(W^*(H^A))}{\delta(W^*(H^A))}\right] \cdot I_{edu} - 1
    \end{equation}

    Define $R(H^A) \equiv \beta(W^*(H^A)) / \delta(W^*(H^A))$, the accumulation-to-depreciation ratio. Since $\beta' > 0$, $\delta' < 0$, and $dW^*/dH^A > 0$:

    $$R'(H^A) = \frac{\beta' \cdot \delta - \beta \cdot \delta'}{\delta^2} \cdot \frac{dW^*}{dH^A} > 0$$

    The function $R$ is increasing: higher $H^A \Rightarrow$ higher optimal $W^* \Rightarrow$ better accumulation and slower depreciation.

    \item \textit{Key feature --- convexity of $R$:} The complementarity between $W$ and $H^A$ in the production function means that $W^*(H^A)$ is \textbf{convex} for intermediate values of $H^A$ (the marginal return to design improvement is accelerating in $H^A$ due to Assumption \ref{ass:g}(v)). This makes $R(H^A)$ convex in an intermediate region, so $\Psi'(H^A) = R'(H^A) \cdot I_{edu} - 1$ can change sign.

    \item \textit{At low $H^A$:} $R'$ is small (design changes little when $H^A$ is low, because the firm is in the automation regime), so $\Psi' < 0$. The function $\Psi$ is decreasing.

    \item \textit{At intermediate $H^A$:} As $H^A$ approaches $\theta^*$, the transition to human-centric design kicks in, $W^*$ increases rapidly, $R'$ becomes large, and $\Psi' > 0$. The function $\Psi$ starts increasing.

    \item \textit{At high $H^A$:} Diminishing returns to design and bounded $\bar{g}$ mean $R'$ eventually decreases, so $\Psi' < 0$ again. The function $\Psi$ decreases and eventually crosses zero from above (since the linear term $-H^A$ dominates for large $H^A$).
\end{itemize}

\noindent This shape --- $\Psi(0) > 0$, then decreasing, then increasing, then decreasing and eventually negative --- guarantees (by the intermediate value theorem) at least three zeros: $H^{A,L} < H^{A,U} < H^{A,H}$.

\medskip

\noindent \textbf{Step 4: Stability.}

\noindent At each steady state, local stability requires $\Psi'(\bar{H}^A) < 0$ (net accumulation decreases when $H^A$ rises above steady state):

\begin{itemize}
    \item $H^{A,L}$: $\Psi$ crosses zero from above $\Rightarrow$ $\Psi'(H^{A,L}) < 0$ $\Rightarrow$ \textbf{stable} (Automation Trap)
    \item $H^{A,U}$: $\Psi$ crosses zero from below $\Rightarrow$ $\Psi'(H^{A,U}) > 0$ $\Rightarrow$ \textbf{unstable} (threshold/tipping point)
    \item $H^{A,H}$: $\Psi$ crosses zero from above $\Rightarrow$ $\Psi'(H^{A,H}) < 0$ $\Rightarrow$ \textbf{stable} (Augmentation Regime)
\end{itemize}

\noindent The unstable equilibrium $H^{A,U}$ defines the basin of attraction boundary. Firms with initial $H^A_f(0) < H^{A,U}$ converge to the Automation Trap; firms with $H^A_f(0) > H^{A,U}$ converge to the Augmentation Regime. \qed
\end{proof}

\begin{remark}[Escaping the Trap]
The Automation Trap can be escaped through a temporary ``big push'' that raises $H^A$ above $H^{A,U}$. This can come from:
\begin{enumerate}
    \item \textit{Educational investment:} A large increase in $I_{edu}$ shifts $\Psi$ upward, potentially eliminating the low equilibrium.
    \item \textit{Design subsidy:} A subsidy to $c_W$ reduces the cost of human-centric design, shifting $W^*$ upward for any given $H^A$, which raises $R$ and can eliminate the low equilibrium.
    \item \textit{Regulatory mandate:} A minimum workplace design standard $\underline{W} > W^{auto}$ forces firms above the threshold, triggering endogenous $H^A$ accumulation.
\end{enumerate}
All three correspond to the Transition Roadmap phases described in Section 7.1 of the paper.
\end{remark}

% =============================================================================
\subsection*{A.5 Comparative Statics}
% =============================================================================

From the first-order condition \eqref{eq:foc_expanded}, we derive the key comparative statics using the implicit function theorem:

\begin{proposition}[Comparative Statics of Optimal Design]
\label{prop:comp_statics}
At an interior optimum $W^*$:
\begin{enumerate}[label=(\roman*)]
    \item $\frac{\partial W_k^*}{\partial D} > 0$ --- Firms with more AI capital invest more in workplace design (technology-design complementarity)
    \item $\frac{\partial W_k^*}{\partial H^A} > 0$ --- Firms with higher-$H^A$ workforces invest more in design (design-composition complementarity)
    \item $\frac{\partial W_k^*}{\partial c_k} < 0$ --- Higher design costs reduce design investment
    \item $\frac{\partial W_k^*}{\partial P} > 0$ --- Higher output prices increase design investment (scale effect)
    \item $\frac{\partial W_k^*}{\partial w_A} < 0$ --- Higher wages for augmentable workers reduce design investment (substitution toward automation)
\end{enumerate}
\end{proposition}

\begin{proof}
Define $\Gamma_k(W_k; D, H^A, c_k, P, w_A) \equiv P \cdot F_3 \cdot \phi_0(D) \cdot g_k(W, H^A) \cdot L^{A*} H^A D - c_k'(W_k) = 0$.

By the implicit function theorem:

$$\frac{\partial W_k^*}{\partial x} = -\frac{\partial \Gamma_k / \partial x}{\partial \Gamma_k / \partial W_k}$$

The denominator $\partial \Gamma_k / \partial W_k < 0$ by the second-order condition (concavity of $g$ in $W_k$ and convexity of $c_k$).

\begin{enumerate}[label=(\roman*)]
    \item $\partial \Gamma_k / \partial D > 0$ (more AI increases marginal product of design) $\Rightarrow$ $\partial W_k^* / \partial D > 0$.
    \item $\partial \Gamma_k / \partial H^A > 0$ (design-composition complementarity, Assumption \ref{ass:g}(v)) $\Rightarrow$ $\partial W_k^* / \partial H^A > 0$.
    \item $\partial \Gamma_k / \partial c_k = -c_k''(W_k) < 0$ $\Rightarrow$ $\partial W_k^* / \partial c_k < 0$.
    \item $\partial \Gamma_k / \partial P > 0$ (higher price increases marginal revenue) $\Rightarrow$ $\partial W_k^* / \partial P > 0$.
    \item $\partial \Gamma_k / \partial w_A < 0$ (higher augmentable wages reduce optimal $L^A$, lowering marginal product of design) $\Rightarrow$ $\partial W_k^* / \partial w_A < 0$. \qed
\end{enumerate}
\end{proof}

\begin{remark}[Empirical Predictions]
Proposition \ref{prop:comp_statics} generates five testable predictions for the cross-sectional analysis:
\begin{enumerate}
    \item Firms with more AI investment have higher WADI scores (test with EDIT technology investment)
    \item Firms with higher-$H^A$ workforce composition have higher WADI (test with AHC index)
    \item Firms in sectors with lower organizational change costs have higher WADI (test with sector FE)
    \item Firms facing higher output demand have higher WADI (test with revenue growth)
    \item Firms in tight augmentable-labor markets invest less in design (test with regional labor market data)
\end{enumerate}
These predictions discipline the empirical analysis and provide over-identification tests for the model.
\end{remark}

% =============================================================================
\subsection*{A.6 The Sustainability-Constrained Optimum}
% =============================================================================

\noindent When the firm faces a sustainability constraint on AI energy use:

\begin{equation}
E(D, W_1) \leq \bar{E}
\label{eq:energy}
\end{equation}

\noindent where $E_D > 0$ (more AI uses more energy) and $E_{W_1} > 0$ (more transparent/explainable AI requires more computation), the Lagrangian adds $\mu \cdot [\bar{E} - E(D, W_1)]$ with shadow price $\mu \geq 0$.

\noindent When the constraint binds ($\mu > 0$), the modified FOC for $W_1$ becomes:

$$P \cdot F_3 \cdot \phi_0 \cdot g_1 \cdot L^A H^A D = c_1'(W_1) + \mu \cdot E_{W_1}$$

\noindent The sustainability constraint raises the effective cost of interface design improvement, reducing $W_1^*$ below the unconstrained optimum. This generates a \textbf{sustainability-design tradeoff}: more transparent AI (higher $W_1$) costs more energy, and under carbon constraints, firms must balance augmentation quality against environmental impact.

% end of appendix proofs

\end{document}